\documentclass[journal, onecolumn, twoside]{IEEEtran}
\normalsize
\ifCLASSINFOpdf
\else
\fi
\usepackage{setspace}
\usepackage{color}
\usepackage{cite}
\usepackage{pifont}

\usepackage{amsmath}
\usepackage{amsfonts}
\usepackage{amssymb}
\usepackage{amsthm}
\usepackage{bm}
\usepackage{mathdots}
\usepackage{cases}

\usepackage{graphicx}
\graphicspath{{../}}
\DeclareGraphicsExtensions{.pdf,.png,.jpg,.jpeg}

\usepackage{epstopdf}
\usepackage{tikz}
\usetikzlibrary{arrows}

\usepackage{booktabs}
\usepackage{multirow}
\usepackage{makecell}
\usepackage{arydshln}
\usepackage{array}
\usepackage{xtab}

\usepackage[caption=true,font=footnotesize]{subfig}

\usepackage{url}
\usepackage{algorithm}
\usepackage{algorithmic}
\usepackage{enumerate}

\allowdisplaybreaks[4]
\usepackage[colorlinks,
            linkcolor=blue,
            anchorcolor=blue,
            citecolor=blue]{hyperref}
\theoremstyle{plain}

\newtheorem{theorem}{Theorem}
\newtheorem{lemma}[theorem]{Lemma}
\newtheorem{proposition}[theorem]{Proposition}
\newtheorem{definition}[theorem]{Definition}
\newtheorem{corollary}[theorem]{Corollary}
\theoremstyle{definition}
\newtheorem{example}{Example}

\newtheorem{remark}{Remark}

\begin{document}
\title{Private Information Retrieval from Joint Systematic MDS-Coded with Non-Colluding Servers: Bounds and Constructions}

\author{Jingke~Xu,~\IEEEmembership{}
        Lirong~Shi,~\IEEEmembership{}
        Peng~Lan,~\IEEEmembership{}
        Weijun~Fang~\IEEEmembership{}
\IEEEcompsocitemizethanks{\IEEEcompsocthanksitem Jingke Xu, Lirong Shi and Peng Lan are with School of Information Science and Engineering,
                    Shandong Agricultural University, Tai'an, 271016, China (emails:xujingke@sdau.edu.cn, 2024110562@sdau.edu.cn, PengLan@sdau.edu.cn).}
\IEEEcompsocitemizethanks{\IEEEcompsocthanksitem Weijun Fang is with the State Key Laboratory of Cryptography and Digital Economy Security, Shandong University, Qingdao, 266237, China, the Key Laboratory of Cryptologic Technology and Information Security, Ministry of Education, Shandong University, Qingdao, 266237, China and the School of Cyber Science and Technology, Shandong University, Qingdao, 266237, China (email: fwj@sdu.edu.cn). {\it (Corresponding Author:  Weijun Fang)}.
}
}
\maketitle

\begin{abstract}
Consider a distributed storage system consisting of $N$ non-colluding servers that collectively store a database of $M$ files encoded using an $[N,K]$ maximum distance separable (MDS) code. A user wishes to retrieve one file privately by accessing the servers without revealing the identity of the requested file. A scheme designed for this purpose is called a joint MDS-coded private information retrieval (PIR) scheme, which was first introduced by Sun and  Tian in 2019 to break the capacity $\frac{1-K/N}{1-(K/N)^M}$ of the separate MDS-coded PIR schemes established by Banawan and Ulukus. However, the capacity of joint MDS-coded PIR remains largely unexplored.

In this paper, we study the capacity of joint MDS-coded PIR with systematic MDS array storage codes under prescribed storage patterns. Specifically, we first derive upper bounds on the capacity of joint MDS-coded PIR for $K=Mt$ and $K=Mt+1$, respectively. These bounds hold for any systematic $(N,K;\ell)$-MDS array storage codes under the storage patterns  $\mathcal{P}=\ell I_M\otimes {\bf 1}_t$ and $\mathcal{P}=\frac{\ell}{M}(MI_M\otimes {\bf 1}_t~\mid {\bf 1}_M^{\intercal})$, respectively.  Moreover, for  $N=K+t$ and $K=Mt$, our upper bound is the first to show the optimality of the joint MDS-coded PIR schemes proposed by Sun and Tian in 2019. We then construct three joint MDS-coded PIR schemes for the cases $N\leq K+t, K=Mt$; $N>K+t, K=Mt$; and $N\leq K+t, K=Mt+1$.
The proposed schemes require small file sizes and achieve higher retrieval rates: the first and third schemes exceed the capacity of separate MDS‑coded PIR schemes, while the second scheme does so when the storage rate $\frac{K}{N}>r_M$ for some  $0<r_M<\frac{M}{M+1}$. In particular, for $K=Mt$ and $N\leq K+t$, the proposed scheme achieves the derived upper bound, thereby establishing that the optimal joint MDS-coded PIR capacity under the considered storage pattern is $1-(1-\frac{1}{M})\frac{K}{N}$. Compared with capacity-achieving separate MDS-coded PIR schemes at the same storage-code rate, the proposed schemes may achieve a substantial relative retrieval-rate improvement: the maximum improvement can exceed $15\%$ when $M\geq 4$, exceed $20\%$ when $M\geq 9$, and asymptotically approach $1-2/e\approx 26.42\%$ as M increases.
\end{abstract}


\section{Introduction}
Introduced by Chor \textit{et al.}~\cite{CKGS95FOCS:PIR} in 1995,
 private information retrieval (PIR) has become a canonical problem  in the study of privacy issues that arise from the retrieval of information from public databases. In the classical PIR model,  a user wishes to privately retrieve
one of $M$ files from a database which contains $N$ non-colluding servers, where each server stores all $M$ files. User privacy needs to be preserved during the retrieval process, requiring that no individual server knows the identity of the desired file. Moreover, PIR has deep connections to cryptography, information theory, and coding theory.

 A central issue in PIR is minimizing the communication cost, which is usually measured by the total number of bits transmitted from the user to the servers (i.e., the upload cost) and from the servers to the user (i.e., the download cost). However, it is more common that each file is quite large in real-world applications, making the download from the servers the dominant communication cost. As a result, the upload cost can be negligible compared to the download cost, and
 the PIR problem was reformulated in~\cite{Sun&Jafar16:CapacityPIR} from the information-theoretic perspective, allowing the user to retrieve arbitrarily large files.  The efficiency of a PIR scheme is measured by the PIR {\it retrieval rate}, that is, the inverse of the download size per unit bit of the desired file, and the {\it capacity} is defined as the supremum of the PIR retrieval rate over all PIR schemes.

Since the capacity was first proposed in~\cite{Sun&Jafar16:CapacityPIR},  determining the capacity for variants of the PIR problem has attracted lots of attention in the literature. Sun and Jafar determined the capacity for non-colluding servers in~\cite{Sun&Jafar16:CapacityPIR} and the capacity for colluding servers in~\cite{Sun&Jafar16:ColludPIR}, respectively. Subsequently, the authors in~\cite{YLW20:CapaPIRCollud} studied the PIR problem under arbitrary collusion patterns and determined its capacity.
 The capacity of PIR for Byzantine  and colluding servers was  presented in~\cite{Bana&Uluk19:BCapacityPIRCoded}, where Byzantine servers are malicious servers that may prevent users from decoding the desired file by sending erroneous answers. Furthermore, many other variations and extensions of the PIR problem have also been studied, such as symmetric PIR~\cite{Sun&Jafar16:SPIR,WS19:SPIR}, multi-file retrieval~\cite{Bana&Uluk18:MPIR}, PIR with the coded databases~\cite{Bana&Uluk16:CapacityPIRCoded,HGHTKK17:PIRMDST,SKHLER19:NonMDSPIR, Xu&Zhang18:SCIS, JieLi,Sun&Jafar18:colludPIR,YZhangGGe17:Coded$T$-PIR},  PIR with small sub-packetization~\cite{
  Xu&Zhang18:SCIS,
ZTSL20:OptimalSubpacketization,  ZYQT20:OptimalSubpacketization}, PIR schemes over small field~\cite{Kale,TCOM22XW,TIT25XF},  the storage cost of PIR~\cite{T20:StoragePIR}, and PIR with side information (or caching) ~\cite{WBU19:PIRSI,KGES20:PIRSI,KARS19:PIRcaching}, a more comprehensive literature survey can be found in~\cite{UAGJ22:SuPIR}.

With the development of distributed storage systems, erasure coding has been widely adopted to improve both storage efficiency
 and failure resistance, especially MDS codes. A PIR from distributed coded storage systems is usually referred to as coded PIR. For an $[N,K]$ MDS-coded storage with  non-colluding servers, the capacity (MDS-coded PIR) was $\frac{1-K/N}{1-(K/N)^M}$, determined by Banawan and Ulukus  in~\cite{Bana&Uluk16:CapacityPIRCoded}.  To reduce the sub-packetization,  each of~\cite{Xu&Zhang18:SCIS, ZTSL20:OptimalSubpacketization, ZYQT20:OptimalSubpacketization} constructed a capacity-achieving MDS-coded PIR scheme, respectively. As pointed out in~\cite{ST19:JMDSPIR},  the storage codes in all these existing works were designed such that each file was independently encoded and stored in the databases and thus could also be recovered individually. Even when the storage code is not an MDS code, the storage system still adopted such independent encoding structure. Although the individual coding structure provides good reliability, it is not the only option in the coded PIR model. Moreover, the authors in~\cite{ST19:JMDSPIR} found that when all the files are stored jointly by using a systematic MDS
 array code in the distributed storage system, the PIR  from such system can  break the capacity~\cite{Bana&Uluk16:CapacityPIRCoded} of  separate MDS-coded PIR in some cases. Specifically,  they constructed the joint MDS-coded PIR schemes for the case of $(M=2, N=nt, K=2t)$ with $n\geq 3$ and the case of $(M\geq 2, N=t(M+1), K=tM)$. The PIR retrieval rates of  these two schemes are  slightly higher than the capacity of  separate MDS-coded PIR. However, the capacity of joint MDS-coded PIR was not theoretically analyzed in~\cite{ST19:JMDSPIR}. This motivates the following natural questions:
  \begin{itemize}
      \item [(1)] For the parameters $(M,N=K+t,K=Mt)$, what is the capacity of the joint systematic MDS-coded PIR?
      \item [(2)] Can one construct joint MDS-coded PIR schemes that achieve retrieval rates higher than the capacity of separate MDS-coded PIR for a broader range of parameters?
  \end{itemize}

\subsection{Contributions }
First, we revisit the capacity problem of joint MDS-coded PIR from a refined storage-model perspective. A joint MDS-coded PIR system should be specified not only by the parameters $(M,N,K)$, but also by a storage pattern $\mathcal P$, which describes how the $M$ files are arranged into the $K$ systematic input blocks, and by a systematic $(N,K;\ell)$ MDS array storage code $\mathcal C$, which encodes these blocks across the $N$ servers. Since the achievable retrieval rate may depend on these two storage ingredients, the capacity should be formulated accordingly.

In Section 2.2, we first introduce the {\it capacity} of joint MDS-coded PIR associated with a fixed storage pattern $\mathcal P$ and a fixed systematic storage code $\mathcal C$, thereby refining the notion in~\cite{ST19:JMDSPIR}. Based on this formulation, we focus on the capacity under a prescribed storage pattern $\mathcal P$, while allowing the storage code $\mathcal C$ to range over all compatible systematic MDS array storage codes. The main contributions of this paper are summarized as follows.

\begin{itemize}
\item[(i)] We establish upper bounds on the capacity of joint MDS-coded PIR for two prescribed storage patterns. Specifically, for $K=Mt$ with $\mathcal{P}=\ell I_M\otimes{\bf 1}_t$, and for $K=Mt+1$ with $\mathcal{P}=\frac{\ell}{M}\big(MI_M\otimes{\bf 1}_t~\mid~{\bf 1}_M^{\intercal}\big)$, we derive upper bounds on $C^{(s)}_{\mbox{\tiny JMDS-PIR}}(M,N,K,\mathcal P)$, as listed in Table~\ref{tab66}. In particular, for the case $K=Mt$, our upper bound provides the previously missing optimality proof for the joint MDS-coded PIR schemes proposed in~\cite{ST19:JMDSPIR}: these schemes attain our upper bound under the storage pattern $\mathcal{P}=\ell I_M\otimes{\bf 1}_t$, and hence are optimal in the corresponding fixed-pattern systematic MDS array storage setting; see Corollary~\ref{cor2}. The proof of the upper bounds relies on characterizing the conditional independence structure of answers, see Lemma~\ref{lem2}, and deriving the recursive inequalities~\eqref{eq12} and~\eqref{eq20} in Lemmas~\ref{lem3} and~\ref{lem4}, respectively. This approach is inspired by the method for determining the capacity of separate MDS-coded PIR in~\cite{Bana&Uluk16:CapacityPIRCoded}, but requires new arguments to handle the joint storage structure.

\item[(ii)] We construct three classes of joint systematic MDS-coded PIR schemes for the cases $N\le K+t$, $K=Mt$; $N>K+t$, $K=Mt$; and $N\le K+t$, $K=Mt+1$, respectively. In particular, for $K=Mt$ and $N\le K+t$, the proposed scheme achieves the upper bound in~\eqref{eq14}. Consequently, we determine the capacity of joint systematic MDS-coded PIR under the storage pattern $\mathcal{P}=\ell I_M\otimes{\bf 1}_t$ as
$C^{(s)}_{\mbox{\tiny JMDS-PIR}}(M,N,K,\mathcal P)
=
1-\left(1-\frac{1}{M}\right)\frac{K}{N}$,
for $K=Mt$ and $N\le K+t$; see Theorem~\ref{thm11}.

\item[(iii)] The proposed schemes exhibit improved performance compared with existing PIR schemes for MDS-coded storage, as summarized in Table~\ref{tab11}. Compared with the known joint schemes in~\cite{ST19:JMDSPIR}, we provide a unified framework for explicitly constructing joint MDS-coded PIR schemes for both $K=Mt$ and $K=Mt+1$, thereby enlarging the range of admissible parameters. For the special case $M=2$, $N=3t$, and $K=2t$, our construction reduces the required finite field size from $4t+4t\binom{3t}{2t}=O(N^{\frac{N}{3}+1})$ in~\cite{ST19:JMDSPIR} to $O(N)$. Compared with the known capacity-achieving separate MDS-coded PIR schemes~\cite{Bana&Uluk16:CapacityPIRCoded,ZTSL20:OptimalSubpacketization,ZYQT20:OptimalSubpacketization}, our proposed joint MDS-coded PIR schemes not only improve the retrieval rate but also reduce the file size. Furthermore, Propositions~\ref{lem15} and~\ref{lem16} show that our schemes achieve a substantial gain in retrieval rate over capacity-achieving separate MDS-coded PIR schemes: at least $15\%$ when $M\ge4$, at least $20\%$ when $M\ge9$, and approaching $1-2/e\approx 26.42\%$ as $M$ increases.
\end{itemize}
\begin{table}[htbp]
	\centering
	\captionsetup{justification=centering}
	\caption{Upper bounds on the capacity of joint MDS-coded PIR for the case of $K=Mt$ and the case of $K=Mt+1$, where the storage code $\mathcal{C}$ is a systematic $(N,K;\ell)$ MDS array code.}
	\resizebox{4in}{!}{
		\renewcommand{\arraystretch}{1.5}
		\begin{tabular}{|c|c|c|c|}
			\hline
			Theorems & {\rm Storage Pattern $\mathcal{P}$} & {\rm Parameters} $(M,N,K)$ & {\rm UB} \\
			\hline
			Thm. \ref{thm1} & $\ell I_M\otimes {\bf 1}_t$ & $K=Mt$ & $1-(1-\frac{1}{M})\frac{K}{N}$ \\ 
			\hline
			Thm. \ref{thm2} & $\frac{\ell}{M}(MI_M\otimes {\bf 1}_t~\mid {\bf 1}_M^{\intercal})$ & $K=Mt+1$ & $\frac{K(N-K+t)}{(K-1)(N-1)}$ \\ 
			\hline
		\end{tabular}
	}
	\label{tab66}
\end{table}
\begin{table}[htbp]
	\centering
	\captionsetup{justification=centering}
	\caption{Comparison with known PIR schemes for MDS-coded servers. Here $r=\frac{K}{N}$ is the rate of the storage code $\mathcal{C}$; ``\textit{vs.}'' compares the scheme's retrieval rate with the capacity of separate MDS-coded PIR; and ``$=$'' and ``$>$'' denote equality and strict improvement, respectively. The notation $>^*$ means that the strict improvement holds when $r>r_M$ for some $0<r_M<M/(M+1)$, where $r_M$ is determined by Proposition~\ref{lem15}(ii).}
	\resizebox{5in}{!}{
		\renewcommand{\arraystretch}{1.5}
		\begin{tabular}{|c|c|c|c|c|c|}
			\hline
			{\rm Reference} & $R_{\rm PIR}$ & \textit{vs.} & File Size $L$ & {\rm Field Size} & {\rm Parameters} $(M,N,K)$ \\
			\hline
			\cite{Bana&Uluk16:CapacityPIRCoded} & $\frac{1-r}{1-r^M}$ & $=$ & $KN^M$ & $N$ & {\rm ALL} \\ 
			\hline
			\cite{ZTSL20:OptimalSubpacketization} & $\frac{1-r}{1-r^M}$ & $=$ & ${\rm lcm}(N-K,K)$ & $N$ & {\rm ALL} \\ 
			\hline
			\cite{ZYQT20:OptimalSubpacketization} & $\frac{1-r}{1-r^M}$ & $=$ & ${\rm lcm}(N-K,K)$ & $N$ & {\rm ALL} \\ 
			\hline
			\cite{ST19:JMDSPIR} & $1-\frac{1}{n}$ & $>$ & $N-t$ & $2t(n-2)(n-1)+2t(n-1)\binom{N}{K}$ & $(2,nt,2t)$ \\ 
			\hline
			\cite{ST19:JMDSPIR} & $\frac{2}{M+1}$ & $>$ & $2t$ & $M(t+1)$ & $(M\geq 3,N=K+t,Mt)$ \\ 
			\hline
			{\bf Construction 1} & $1-(1-\frac{1}{M})r$ & $>$ & $t(N-K+t)$ & $N$ & $(M,N\leq K+t,Mt)$ \\ 
			\hline
			{\bf Construction 2} & $\frac{(N+K-t)}{NM}$ & $>^*$ & $t(N+K-t)$ & $N$ & $(M,N> K+t,Mt)$ \\ 
			\hline
			{\bf Construction 3} & $\frac{K(N-K+t)}{M((t+1)N-K)}$ & $>$ & $K \frac{{\rm lcm}(M,N-K+t)}{M}$ & $N$ & $(M,N\leq K+t,Mt+1)$ \\ 
			\hline
		\end{tabular}
	}
	\label{tab11}
\end{table}

\subsection{Organization}

The rest of this paper is organized as follows. Section 2 provides a precise description of the system model. In Sections 3 and 4, we provide upper bounds on capacity and present some constructions of joint systematic MDS-coded PIR for two classes of parameters, respectively.  Finally, Section 5 concludes  this paper.

\section{System Model}\label{sec2}
In this section, we first introduce the joint coded distributed storage system,  and then provide a formal description of the joint MDS-coded PIR model.

For any two positive integers $a\leq b$, let $[a:b]$ be the set $\{a,a+1,...,b\}$ and $[n]\triangleq[1:n]$. Throughout the paper, we usually use capital letters to denote matrices (e.g. $A,B$) and bold lowercase letters to represent vectors (e.g. ${\bf a,b}$), and $A^{\intercal}$, ${\bf a}^{\intercal}$ denote the transpose of matrix $A$ and vector ${\bf a}$, respectively. For a block vector ${\bf v}=({\bf v}_1,{\bf v}_2,\cdots,{\bf v}_n)$ and $\Lambda=\{i_1,i_2,\cdots,i_s\}\subseteq [n]$, let ${\bf v}_\Lambda=({\bf v}_{i_1},{\bf v}_{i_2},\cdots,{\bf v}_{i_s})$. For an $m\times n$ matrix $A=(a_{i,j})_{i\in[m],j\in[n]}$, let $A(\Lambda,\Gamma)$ be the matrix $(a_{i,j})_{i\in\Lambda,j\in\Gamma}$, whose rows and columns are labeled by the subset  $\Lambda\subseteq[m]$ and $\Gamma\subseteq [n]$, respectively.  For an $m\times n$ matrix $A=(a_{i,j})_{i\in[m],j\in[n]}$ and a $p\times t$ matrix $B$, the Kronecker product $A\otimes B=(a_{i,j}B)_{i\in[m],j\in[n]}$, which is an $mp\times nt$  matrix.
 In particular,  it is respectively abbreviated as $A(:,\Gamma)$ or $A(\Lambda,:)$ if $\Lambda=[m]$ or $\Gamma=[n]$. Moreover, a block matrix formed as $A=(A^{(1)},A^{(2)},...,A^{(n)})$ is defined by concatenation with matrices having the same number of rows.
For the above matrix $A$ and a subset $\Gamma=\{i_1,\cdots,i_s\}\subseteq [n]$, denote $A^\Gamma=(A^{(i_1)},\cdots,A^{(i_s)})$.
In the absence of specific instructions, vectors throughout this paper are always referred to row vectors.

\subsection{Systematic Joint MDS-Coded Distributed Storage System}

Suppose a distributed storage system consists of $N$ servers ${\rm Serv}^{(1)}, \cdots{\rm Serv}^{(N)}$, which adopts a systematic $(N,K;\ell)$ MDS array code to store a message vector $({\bf c}_1,\cdots,{\bf c}_K)$ that is composed of $M $ files $W_1,\cdots,W_M$. These files are independent and each consists of $L$ symbols drawn independently and uniformly from the finite field $\mathbb{F}_q$, that is,
\begin{equation}\label{eq1}
\forall\ i\in[M],~H(W_i)=L,
H(W_1,W_2,\cdots,W_M)=ML,
\end{equation}
where $H(\cdot)$ is the entropy function with base $q$. Specifically, a systematic $(N,K;\ell)$  MDS array code $\mathcal{C}$  encodes  the message vector $({\bf c}_1,\cdots,{\bf c}_K)$ to a codeword ${\bf c}=({\bf c}_1,{\bf c}_2,\cdots, {\bf c}_N)$ by using a systematic generator matrix $G$, i.e.,
\begin{equation}\label{eq2}
{\bf c}=({\bf c}_1,{\bf c}_2,\cdots, {\bf c}_N)=({\bf c}_1,{\bf c}_2,\cdots, {\bf c}_K)G,
\end{equation} where ${\bf c}_i\in\mathbb{F}_q^\ell$, and the MDS property means that for any $K$-subset $\Gamma$ of $[N]$, the codeword ${\bf c}$ can be recovered by  ${\bf c}_\Gamma$.
Moreover, each server ${\rm Serv}^{(j)}$ stores the coded data ${\bf c}_j$ for $j\in[N]$. So, we call such a storage system an {\it $(N,K;\ell)$  joint systematic MDS-coded distributed storage system}.

 Now, we introduce  {\it storage pattern}  to characterize the relationship between  $M$ files and the message vector $({\bf c}_1,\cdots,{\bf c}_K)$.
 Without loss of generality, we may assume that $K\mid ML$, or, $ML=K\ell$ for some integer $\ell$.
\begin{definition}
A matrix $\mathcal{P}=(p_{i,j})\in \mathbb{N}^{M\times K}$ is called as a {\it storage pattern} of an $(N,K;\ell)$ joint systematic MDS-coded distributed storage system that stores  $M$ files $W_1,\cdots,W_M\in\mathbb{F}_q^L$, if  $W_i=({\bf W}_{i,1},{\bf W}_{i,2},\cdots,{\bf W}_{i,K})$ and ${\bf c}_{j}=({\bf W}_{1,j},{\bf W}_{2,j},\cdots,{\bf W}_{M,j})\in \mathbb{F}_q^\ell$ for all $i\in[M],j\in[K]$, where each ${\bf W}_{i,j} \in\mathbb{F}_q^{p_{i,j}}$ is the part of $W_i$ stored on the ${\rm Serv}^{(j)}$.
\end{definition}
 For any storage pattern $\mathcal{P}$, it has that
$
\begin{cases}
\sum^K_{j=1}p_{i,j}=L, & i\in[M],\\
\sum^{M}_{i=1}p_{i,j}=\ell, &j\in[K].
\end{cases}
$ Now, we use an example to explain the storage code and storage pattern.

\begin{example}\label{ex0} Sun and Tian in~\cite{ST19:JMDSPIR} constructed two classes of  joint MDS-coded PIR schemes. For the case of  $(M=2,N=nt,K=2t)$ with $n\geq 3$, $L=t(n-1)$ and $\ell=n-1$, the storage code is a systematic $(N=nt,K=2t;\ell=n-1)$ MDS array code and the storage pattern is $\mathcal{P}= (n-1)I_2\otimes {\bf 1}_t \in \mathbb{N}^{2\times K}$ in their scheme.  For the case of $(M\geq 2,N=t(M+1),K=Mt)$ with $L=2t$ and $\ell=2$,  the storage code is a systematic $(N=t(M+1),K=Mt;\ell=2)$ MDS array code and
the storage pattern $\mathcal{P}=2 I_M\otimes {\bf 1}_t$. One can find that the storage pattern and dimension $K$ have the same form in these two cases, that is, $\mathcal{P}=\ell I_M\otimes {\bf 1}_t$  and $K=Mt$ for $M\geq 2, t\geq 1$.
 \end{example}
 Due to the practical utility of systematic codes, we focus on the joint PIR problem within the systematic MDS array code framework, and then explore  the affirmative answers to questions (1) and (2).

\subsection{Joint Systematic MDS-coded PIR Model}
 A PIR scheme allows a user to retrieve a file, say $W_\theta$, for some $\theta\in [M]$ by accessing the $N$ servers while ensuring the secrecy of the index $\theta$ against any individual server. A  PIR scheme over a joint systematic MDS-coded distributed storage system can be formally described by the following two phases:
\begin{itemize}
  \item {\bf Query phase.} Suppose a user wants to retrieve a file $W_\theta$ for some $\theta\in[M]$, then the user privately generates queries ${\rm
      Que}(\theta,\mathcal{S})=(Q_\theta^{(1)},...,Q_\theta^{(N)})$, and sends $Q_\theta^{(i)}$ to ${\rm Serv}^{(i)}$, $1\leq i\leq N$, where $\mathcal{S}$ are some random resources. Note that $\mathcal{S}$ and $\theta$ are private information only known by the user, and the {\it query function} ${\rm Que}(\cdot,\cdot)$ is determined by the scheme.  For simplicity, we define the query set $\mathcal{Q}=\{Q_\theta^{(j)}, \mathcal{S}: \theta\in[M],j\in[N]\}$. Hence, it is natural to assume the user has no  prior knowledge of all files, that is,
      \begin{equation}\label{eq3}
      I(W_{[M]};\mathcal{Q})=0,
      \end{equation}
      where $I(X;Y)=H(X)-H(X|Y)$, which is the mutual information function with base $q$.
  \item {\bf Response phase.} After receiving  the query $Q_\theta^{(i)}$, $1\leq i\leq N$, the $i$th server ${\rm Serv}^{(i)}$  computes ${\rm Ans}^{(i)}(Q_\theta^{(i)},{\bf c}_i)=A_\theta^{(i)}$ and sends it back to the user, where ${\rm Ans}^{(i)}(\cdot,\cdot)$ is the {\it answer  function} defined by the scheme.  Hence,
      \begin{equation}\label{eq4}
      H(A^{(i)}_\theta\mid Q^{(i)}_\theta,{\bf c}_i)=0.
      \end{equation}
\end{itemize}

  Moreover, the functions ${\rm Que}$ and ${\rm Ans}^{(i)}$ are deterministic in the PIR scheme and need to satisfy the following two conditions.
\begin{itemize}
\item[(1)]{\it Correctness: } The user can definitely recover the file $W_\theta$ after collecting all answers from all servers. That is,
\begin{equation}\label{eq5}
 H(W_\theta\mid A^{[N]}_\theta,Q^{[N]}_\theta)=0.
\end{equation}
\item[(2)]{\it Privacy:} For any $n\in [N]$, the server ${\rm Serv}^{(n)}$  should be indistinguishable from the scheme. So, for $\theta\in[M]$,
    $$(Q^{(n)}_\theta, A^{(n)}_{\theta},{\bf c}_n,W_{[M]})\sim (Q^{(n)}_1, A^{(n)}_{1},{\bf c}_n,W_{[M]}),$$where $A\sim B$ denotes that random variables $A$ and $B$ are identically distributed. This implies that the queries and answers are independent of the desired index $\theta$ and ${\rm Serv}^{(n)}$ gets no information on the index $\theta$.  That is, for $n\in[N]$,
\begin{equation}\label{eq6}
I(\theta;A_\theta^{(n)},Q_\theta^{(n)},W_{[M]})=0.
\end{equation}
\end{itemize}

To measure the download efficiency of a joint MDS-coded PIR scheme, we define
its {\it retrieval rate} as
$R_{\rm PIR}=\frac{H(W_\theta)}{\sum_{n=1}^{N}H(A^{(n)}_\theta)}=\frac{L}{D},$ where $D=\sum_{n=1}^{N}H(A^{(n)}_\theta)$.
This quantity characterizes the number of bits of desired information that can be retrieved per bit of total downloaded.
For a fixed storage pattern $\mathcal{P}$ and a fixed systematic storage code $\mathcal{C}$, the {\it capacity} of joint MDS-coded PIR, denoted by $C_{\mbox{\tiny JMDS-PIR}}(M, N, K,\mathcal{P},\mathcal{C})$, is defined as the supremum of achievable retrieval rates over all joint MDS-coded PIR schemes operating on the corresponding storage system.
Furthermore, for a fixed storage pattern $\mathcal P$, we define
$C^{(s)}_{\mbox{\tiny JMDS-PIR}}(M,N,K,\mathcal P)=\max_{\mathcal C}C_{\mbox{\tiny JMDS-PIR}}(M,N,K,\mathcal P,\mathcal C),$
where the maximum is taken over all systematic $(N,K;\ell)$ MDS array storage codes.

In this paper, we aim to establish upper bounds on the capacity  $C^{(s)}_{\mbox{\tiny JMDS-PIR}}(M,N,K,\mathcal{P})$, and construct the corresponding schemes whose rates are strictly greater than the capacity $C_\oplus(M,N,K)$ of separate MDS-coded PIR for the cases $K=Mt$  with storage patterns $\mathcal{P}=\ell I_M\otimes {\bf 1}_t$ and  $K=Mt+1$ with $\mathcal{P}=\frac{\ell}{M}(MI_M\otimes {\bf 1}_t~\mid {\bf 1}_M^{\intercal})$, respectively.

\section{Upper Bounds on the Capacity of Joint Systematic MDS-coded PIR for Two Cases }\label{sec3}

Before deriving the upper bound on  the capacity, we first build some key lemmas.

\begin{lemma}\label{lem1} For a joint systematic  MDS-coded PIR scheme, for any $\theta,\theta'\in [M]$, any subset $\Lambda\subseteq[M]$, and $i\in[N]$,
\begin{equation}\label{eq7}
H(A^{(i)}_\theta|W_\Lambda,\mathcal{Q})=H(A^{(i)}_{\theta'}|W_\Lambda,\mathcal{Q}).
\end{equation}
\end{lemma}

\begin{proof} By the privacy \eqref{eq6}, it has that for $i\in[N],\Lambda\subseteq [M]$, $I(\theta;A^{(i)}_\theta,Q^{(i)}_\theta,W_\Lambda)=0$. Hence,   for any $\theta,\theta'\in [M]$, $H(A^{(i)}_\theta|W_\Lambda,Q^{(i)}_\theta)=H(A^{(i)}_{\theta'}|W_\Lambda,Q^{(i)}_{\theta'})$.

Next it is sufficient to show that for all $i\in[N],\theta \in[M]$, $H(A^{(i)}_\theta|W_\Lambda,\mathcal{Q})=H(A^{(i)}_\theta|W_\Lambda,Q^{(i)}_\theta)$. Note that $A^{(i)}_\theta$ is a deterministic function of  $({\bf c}_i,Q^{(i)}_{\theta})$, hence $A^{(i)}_\theta$  is conditionally independent of the  $\mathcal{Q}\setminus \{Q^{(i)}_{\theta}\}$ given $Q^{(i)}_{\theta}$. That is, $I(A^{(i)}_\theta; \mathcal{Q}\setminus \{Q^{(i)}_{\theta}\}|Q^{(i)}_{\theta})=0$, which implies that for $\Lambda\subseteq[M]$,  $I(A^{(i)}_\theta; \mathcal{Q}\setminus \{Q^{(i)}_{\theta}\}|W_\Lambda,Q^{(i)}_{\theta})=0$. So,
\begin{align*}
H(A^{(i)}_\theta|W_\Lambda,\mathcal{Q})&=H(A^{(i)}_\theta|W_\Lambda,\mathcal{Q})+I(A^{(i)}_\theta; \mathcal{Q}\setminus \{Q^{(i)}_{\theta}\}|W_\Lambda,Q^{(i)}_{\theta})=H(A^{(i)}_\theta|W_\Lambda,Q^{(i)}_\theta).
\end{align*}\end{proof}

Next, we characterize the subset $\Gamma$ of $[N]$ that contains independent answers, given the data ${\bf c}_\Lambda$ stored in any subset of servers and queries $\mathcal{Q}$ in joint MDS PIR schemes.
\begin{lemma}\label{lem2}  For a joint systematic  MDS-coded PIR scheme, for $\theta\in [M]$ and  $\Lambda\subsetneq[N]$ with $|\Lambda|<K$, and $\Gamma\subseteq[N]\setminus\Lambda$ with $|\Gamma|=K-|\Lambda|$,
\begin{equation}\label{eq8}
H(A^{\Gamma}_\theta|{\bf c}_{\Lambda},\mathcal{Q})=\sum_{i\in\Gamma}H(A^{(i)}_\theta|{\bf c}_{\Lambda},\mathcal{Q}).
\end{equation}
\end{lemma}
\begin{proof} Note that  $\Gamma\cap \Lambda=\emptyset$,  so it has that $|\Gamma\cup \Lambda|=K$. By the MDS property of the storage code, we know that $H({\bf c}_\Gamma,{\bf c}_\Lambda)=H(W_{[M]})=ML$,  which implies that ${\bf c}_\Gamma,{\bf c}_\Lambda$ are  independent. Hence, ${\bf c}_\Gamma$ are also  conditionally independent given ${\bf c}_\Lambda$.
 Moreover, $A^{(i)}_\theta,i\in\Gamma$ are  deterministic functions of  $({\bf c}_i,Q^{(i)}_{\theta})$, which implies that  they are conditionally independent given ${\bf c}_\Lambda, \mathcal{Q}$. Hence, the proof is completed.
\end{proof}

Combining Lemma \ref{lem1} and Lemma \ref{lem2}, one can directly obtain the following corollary.

\begin{corollary} For a joint systematic  MDS-coded PIR scheme, if there is a subset $\Lambda\subsetneq [N]$ with $|\Lambda|<K$ such that ${\bf c}_\Lambda=W_{\Lambda'}$ for some fixed $\Lambda'\subseteq [M]$, then for $\theta,\theta'\in [M]$  and $\Gamma\subseteq[N]\setminus\Lambda$ with $|\Gamma|=K-|\Lambda|$,
\begin{equation}\label{eq9}
H(A^{\Gamma}_\theta|{\bf c}_{\Lambda},\mathcal{Q})=H(A^{\Gamma}_{\theta'}|{\bf c}_{\Lambda},\mathcal{Q}).
\end{equation}
\end{corollary}

So far, the preparation for deriving the upper bound on the capacity has been completed. In next two subsections, we present the upper bounds of $C^{(s)}_{\mbox{\tiny JMDS-PIR}}(M,N,K,\mathcal{P})$ for $K=Mt$, $\mathcal{P}=\ell I_M\otimes {\bf 1}_t$, and $K=Mt+1,~\mathcal{P}=\frac{\ell}{M}(MI_M\otimes {\bf 1}_t~\mid {\bf 1}_M^{\intercal})$, respectively.

\subsection{The upper bound for $K=Mt$ and $\mathcal{P}=\ell I_M\otimes {\bf 1}_t$}

For the storage pattern $\mathcal{P}=\ell I_M\otimes {\bf 1}_t$, one can know that
$W_i=({\bf c}_{(i-1)t+1},{\bf c}_{(i-1)t+2},\cdots,{\bf c}_{it})$ for all $i\in[M]$. Denote $\chi_\Lambda=\{(i-1)t+j:i\in\Lambda,j\in[t]\}$. So, it has that for all  $\Lambda \subseteq [M]$,
 ${\bf c}_{\chi_\Lambda}=W_\Lambda$.
Moreover, for all
$\theta\in[M]$ and $\Lambda\subseteq[M]$,
\begin{equation}\label{eq11}
H(A^{\chi_\Lambda}_\theta|W_\Lambda,\mathcal{Q})=H(A^{\chi_\Lambda}_\theta|{\bf c}_{\chi_\Lambda},\mathcal{Q})=0.
\end{equation}

Next, we establish a key lemma to derive the upper bound of the capacity for the case of $K=Mt$ and the storage pattern $\mathcal{P}=\ell I_M\otimes {\bf 1}_t$.
\begin{lemma}\label{lem3}  For a joint systematic MDS-coded PIR scheme with the storage pattern matrix $\mathcal{P}=\ell I_M\otimes {\bf 1}_t$, for any subset $\Lambda\subsetneq[M],\theta\in \Lambda$, and any $\theta'\in [M]\setminus\Lambda$,
\begin{equation}\label{eq12}
H(A^{[N]}_\theta|{\bf c}_{\chi_\Lambda},\mathcal{Q})\geq \frac{K-|\Lambda|t}{N-|\Lambda|t}L+ \frac{K-|\Lambda|t}{N-|\Lambda|t} H(A^{[N]}_{\theta'}|{\bf c}_{\chi_{\Lambda\cup \{\theta'\}}},\mathcal{Q}).
\end{equation}
\end{lemma}
\begin{proof}By \eqref{eq11}, it has that $H(A_\theta^{\chi_\Lambda}|{\bf c}_{\chi_\Lambda},\mathcal{Q},A_\theta^{[N]\setminus\chi_\Lambda})=0$.
Then, \begin{equation}\label{eq13}H(A^{[N]}_\theta|{\bf c}_{\chi_\Lambda},\mathcal{Q})=H(A_\theta^{[N]\setminus\chi_\Lambda}|{\bf c}_{\chi_\Lambda},\mathcal{Q}) +H(A_\theta^{\chi_\Lambda}|{\bf c}_{\chi_\Lambda},\mathcal{Q},A_\theta^{[N]\setminus\chi_\Lambda})=H(A_\theta^{[N]\setminus\chi_\Lambda}|{\bf c}_{\chi_\Lambda},\mathcal{Q}).\end{equation}
According to \eqref{eq9} and $K-|\Lambda|t$-subset $\Gamma\subseteq[N]\setminus\chi_\Lambda$, it has that $H(A_\theta^{[N]\setminus\chi_\Lambda}|{\bf c}_{\chi_\Lambda},\mathcal{Q})\geq H(A^\Gamma_\theta|{\bf c}_{\chi_\Lambda},\mathcal{Q})=H(A^\Gamma_{\theta'}|{\bf c}_{\chi_\Lambda},\mathcal{Q}).$ By Han's inequality \cite[Thm. 17.6.1]  {CT12:EIT} and  \eqref{eq13},
$$H(A^{[N]}_\theta|{\bf c}_{\chi_\Lambda},\mathcal{Q})
\geq \frac{1}{\binom{N-|\Lambda|t}{K-|\Lambda|t}}
\sum_{\substack{\Gamma:|\Gamma|=K-|\Lambda|t,\\\Gamma\subseteq[N]\setminus\chi_\Lambda}}
H(A^\Gamma_{\theta'}|{\bf c}_{\chi_\Lambda},\mathcal{Q})
\geq\frac{K-|\Lambda|t}{N-|\Lambda|t}H(A_{\theta'}^{[N]\setminus\chi_\Lambda}|{\bf c}_{\chi_\Lambda},\mathcal{Q})=\frac{K-|\Lambda|t}{N-|\Lambda|t}H(A^{[N]}_{\theta'}|{\bf c}_{\chi_\Lambda},\mathcal{Q}).$$
Moreover, $H(A^{[N]}_{\theta'}|{\bf c}_{\chi_\Lambda},\mathcal{Q})=H(W_{\theta'}|{\bf c}_{\chi_\Lambda},\mathcal{Q})-H(W_{\theta'}|{\bf c}_{\chi_\Lambda},\mathcal{Q},A^{[N]}_{\theta'})+H(A^{[N]}_{\theta'}|{\bf c}_{\chi_{\Lambda\cup\{\theta'\}}},\mathcal{Q})=L+H(A^{[N]}_{\theta'}|{\bf c}_{\chi_{\Lambda\cup\{\theta'\}}},\mathcal{Q}).$ Substituting this into the above inequality to obtain \eqref{eq12}. \end{proof}

\begin{theorem}\label{thm1} For $K=Mt$ and the storage pattern matrix $\mathcal{P}=\ell I_M\otimes {\bf 1}_t$,
the capacity of the  $(N,K;\ell)$ joint systematic MDS-coded PIR satisfies
\begin{equation}\label{eq14}
C^{(s)}_{\mbox{\tiny\rm JMDS-PIR}}(M,N,K,\mathcal{P})\leq 1-(1-\frac{1}{M})\frac{K}{N}.
  \end{equation}
\end{theorem}
\begin{proof}To prove \eqref{eq14}, it is sufficient to show that for any joint MDS-coded PIR scheme with storage code  $\mathcal{C}$ and storage pattern $\mathcal{P}$, its retrieval rate
 $R_{\rm PIR}=\frac{L}{D}\leq 1-(1-\frac{1}{M})\frac{K}{N}.$

For $\theta\in[M]$, we prove that \begin{equation}\label{eq15}H(A^{[N]}_{\theta}|\mathcal{Q})\geq L(1+\sum^{M-1}_{j=1}\prod^j_{i=1}\frac{ K-it}{ N-it}).\end{equation}

First, we have
\begin{equation}\label{eq16}
L\stackrel{{\rm (i)}}{=}H(W_\theta|\mathcal{Q})-H(W_\theta|\mathcal{Q},A^{[N]}_\theta)
=I(W_\theta;A^{[N]}_\theta|\mathcal{Q})
\stackrel{{\rm (ii)}}{=}H(A^{[N]}_\theta|\mathcal{Q})-H(A^{[N]}_\theta|{\bf c}_{\chi_{\{\theta\}}},\mathcal{Q}),
\end{equation}
where (i) follows from \eqref{eq3} and \eqref{eq5}, (ii) is due to \eqref{eq11}. Then by \eqref{eq12}, $H(A^{[N]}_\theta|{\bf c}_{\chi_{\{\theta\}}},\mathcal{Q})\geq \frac{K-t}{N-t}(L+H(A^{[N]}_{\theta'}|{\bf c}_{\chi_{\{\theta,\theta'\}}},\mathcal{Q}))$. By recursively using \eqref{eq12} in Lemma \ref{lem3}, we have
\begin{align}
H(A^{[N]}_\theta|{\bf c}_{\chi_{\{\theta\}}},\mathcal{Q})\geq&  L\sum^{M-1}_{j=1}\prod^j_{i=1}\frac{ K-it}{ N-it}+\prod^{M-1}_{i=1}\frac{K-it}{ N-it}H(A^{[N]}_{\theta''}|{\bf c}_{\chi_{[M]}},\mathcal{Q})
\stackrel{{\rm (i)}}{=}L(\sum^{M-1}_{j=1}\prod^j_{i=1}\frac{ K-it}{ N-it}), \label{eq17}
\end{align}
where (i) follows from \eqref{eq4} and ${\bf c}_{\chi_{[M]}}={\bf c}_{[K]}=(W_1,\cdots,W_M)$. Then, \eqref{eq15} can be directly obtained by \eqref{eq16} and \eqref{eq17}. So,
\begin{equation}\label{eq010}\begin{split}
\frac{L}{D}\leq& \frac{L}{H(A^{[N]}_\theta|\mathcal{Q})}\leq (1+\sum^{M-1}_{j=1}\prod^j_{i=1}\frac{ K-it}{ N-it})^{-1}\\
=&(1+\sum^{M-3}_{j=1}\prod^j_{i=1}\frac{ K-it}{ N-it}+(\ \prod^{M-2}_{i=1}\frac{ K-it}{ N-it})(1+\frac{K-(M-1)t}{N-(M-1)t}))^{-1}\\
=&(1+\sum^{M-4}_{j=1}\prod^j_{i=1}\frac{ K-it}{ N-it}+(\prod^{M-3}_{i=1}\frac{ K-it}{N-it})(\frac{N-(M-3)t}{N-(M-1)t}))^{-1}\\
=&(1+\frac{K-t}{N-t}(1+\frac{K-2t}{N-(M-1)t}))^{-1}\\
=&1-(1-\frac{1}{M})\frac{K}{N}.
\end{split}
\end{equation}
\end{proof}

Recall that the schemes for $(M=2,N=nt,K=2t)$ and for $(M,N=(M+1)t,K=Mt)$ constructed in~\cite{ST19:JMDSPIR} have retrieval rates $\frac{n-1}{n}$ and $\frac{2}{M+1}$   respectively, which match the upper bound \eqref{eq14}. Thus, the schemes constructed in~\cite{ST19:JMDSPIR} are optimal for the storage pattern $\mathcal{P}=\ell I_M\otimes {\bf 1}_t$. Moreover, we obtain the following corollary:
\begin{corollary}\label{cor2} For the case of $M=2, N=nt, K=2t$ with $n\geq 3$ and the storage pattern $\mathcal{P}=(n-1) I_2\otimes {\bf 1}_t$,
\begin{equation*}\label{eq18}
C^{(s)}_{\mbox{\tiny\rm JMDS-PIR}}(M=2,N=nt,K=2t,\mathcal{P})= 1-(1-\frac{1}{M})\frac{K}{N}=1-\frac{1}{n}.
  \end{equation*}
For the  case of $M\geq 2,N=(M+1)t,K=Mt$ and  the storage pattern $\mathcal{P}=2 I_M\otimes {\bf 1}_t$,
\begin{equation*}\label{eq18}
C^{(s)}_{\mbox{\tiny\rm JMDS-PIR}}(M,N=(M+1)t,K=Mt,\mathcal{P})= 1-(1-\frac{1}{M})\frac{K}{N}=\frac{2}{M+1}.
  \end{equation*}
\end{corollary}

\subsection{The upper bound for $K=Mt+1$ and $\mathcal{P}=\frac{L}{K}(MI_M\otimes {\bf 1}_t~{\bf 1}_M^{\intercal})$}

For the storage pattern $\mathcal{P}=\frac{L}{K}(MI_M\otimes {\bf 1}_t~{\bf 1}_M^{\intercal})$, one can know that   $W_i=({\bf c}_{(i-1)t+1},{\bf c}_{(i-1)t+2},\cdots,{\bf c}_{it},{\bf W}_{i,K})$ for all $i\in [M]$, and ${\bf c}_K=({\bf W}_{1,K},\cdots,{\bf W}_{M,K})$. Thus, it has  that for all
$\theta\in[M],\Lambda\subseteq[M]$,
\begin{equation}\label{eq19}
H(A^{\chi_\Lambda}_\theta|W_\Lambda,\mathcal{Q})=H(A^{\chi_\Lambda}_\theta|{\bf c}_{\chi_\Lambda},{\bf    c}_{K,\Lambda},\mathcal{Q})=0,
\end{equation}
 where $\chi_\Lambda=\{(i-1)t+j:i\in\Lambda,j\in[t]\}$.

\begin{lemma}\label{lem4}  For a joint systematic MDS-coded PIR scheme with the storage pattern $\mathcal{P}=\frac{L}{K}(MI_M\otimes {\bf 1}_t~{\bf 1}_M^{\intercal})$, for any subset $\Lambda\subsetneq[M],\theta\in \Lambda$, and any $\theta'\in [M]\setminus\Lambda$,
\begin{equation}\label{eq20}
H(A^{[N]}_\theta|{\bf c}_{\chi_\Lambda},{\bf c}_K,\mathcal{Q})\geq \frac{K-1-|\Lambda|t}{N-1-|\Lambda|t}(L-\frac{L}{K})+ \frac{K-1-|\Lambda|t}{N-1-|\Lambda|t} H(A^{[N]}_{\theta'}|{\bf c}_{\chi_{\Lambda\cup \{\theta'\}}},{\bf c}_K,\mathcal{Q}).
\end{equation}
\end{lemma}
\begin{proof}
By \eqref{eq19}, we have \begin{equation}\label{eq21}H(A^{[N]}_\theta|{\bf c}_{\chi_\Lambda},{\bf c}_K,\mathcal{Q})
=H(A_\theta^{[N]\setminus(\chi_\Lambda\cup\{K\})}|{\bf c}_{\chi_\Lambda},{\bf c}_K,\mathcal{Q}).\end{equation}
Moreover, it has that $$H(A_\theta^{[N]\setminus(\chi_\Lambda\cup\{K\})}|{\bf c}_{\chi_\Lambda},{\bf c}_K,\mathcal{Q})\geq \frac{1}{\binom{N-1-|\Lambda|t}{K-1-|\Lambda|t}}
\sum_{\substack{\Gamma:|\Gamma|=K-1-|\Lambda|t,\\\Gamma\subseteq[N]\setminus(\chi_\Lambda\cup\{K\})}}
H(A^\Gamma_{\theta}|{\bf c}_{\chi_\Lambda},{\bf c}_K,\mathcal{Q}).$$
Combining $ H(A^\Gamma_\theta|{\bf c}_{\chi_\Lambda}, {\bf c}_K,\mathcal{Q})=H(A^\Gamma_{\theta'}|{\bf c}_{\chi_\Lambda}, {\bf c}_K,\mathcal{Q})$, \eqref{eq21} and Han's inequality \cite[Thm. 17.6.1]{CT12:EIT}, it holds that
 \begin{align*}
H(A^{[N]}_\theta|{\bf c}_{\chi_\Lambda},{\bf c}_K,\mathcal{Q})
\geq&
\frac{K-1-|\Lambda|t}{N-1-|\Lambda|t}H(A_{\theta'}^{[N]\setminus(\chi_\Lambda\cup\{K\})}|{\bf c}_{\chi_\Lambda},{\bf c}_K,\mathcal{Q})\\
=&\frac{K-1-|\Lambda|t}{N-1-|\Lambda|t}(H(W_{\theta'}|{\bf c}_{\chi_\Lambda},{\bf c}_K,\mathcal{Q})-
H(W_{\theta'}|{\bf c}_{\chi_{\Lambda}},{\bf c}_{K},\mathcal{Q},A^{[N]}_{\theta'})
+H(A^{[N]}_{\theta'}|{\bf c}_{\chi_{\Lambda\cup\{\theta'\}}},{\bf c}_{K},\mathcal{Q}))\\
\stackrel{{\rm (i)}}{=}&\frac{K-1-|\Lambda|t}{N-1-|\Lambda|t}(L-\frac{L}{K})+ \frac{K-1-|\Lambda|t}{N-1-|\Lambda|t} H(A^{[N]}_{\theta'}|{\bf c}_{\chi_{\Lambda\cup\{\theta'\}}},{\bf c}_K,\mathcal{Q}),
\end{align*}
where (i)  is due to the fact that $H(W_{\theta'}|{\bf c}_{\chi_\Lambda},{\bf c}_K,\mathcal{Q})=H(W_{\theta'}|{\bf W}_{\theta',K})=H(W_{\theta'})-H({\bf W}_{\theta',K})=L-\frac{L}{K}$, which follows from \eqref{eq5}, \eqref{eq3} and \eqref{eq1}.\end{proof}

Now, let us build the upper bound of capacity for this case.
\begin{theorem}\label{thm2}
For $K=Mt+1$ and the storage pattern matrix $\mathcal{P}=\frac{L}{K}(MI_M\otimes {\bf 1}_t~{\bf 1}_M^{\intercal})$,
\begin{equation}\label{eq22}
C^{(s)}_{\mbox{\tiny\rm JMDS-PIR}}(M,N,K,\mathcal{P})\leq \frac{K(N-K+t)}{(K-1)(N-1)},
  \end{equation}
\end{theorem}

\begin{proof}
 By the definition of the capacity of joint systematic MDS-coded PIR, it is sufficient to show that for any joint MDS-coded PIR scheme with systematic storage code  $\mathcal{C}$ and storage pattern $\mathcal{P}$, its retrieval rate
 $R_{\rm PIR}=\frac{L}{D}\leq \frac{K(N-K+t)}{(K-1)(N-1)}.$

For $\theta\in[M]$, we prove that \begin{equation}\label{eq23}H(A^{[N]}_{\theta}|{\bf c}_K,\mathcal{Q})\geq L(1-\frac{1}{K}) \frac{N-1}{N-K+t}.\end{equation}

Similarly to \eqref{eq16},
according to  \eqref{eq5} and \eqref{eq19},
 one can obtain that \begin{equation}\label{eq24}
L(1-\frac{1}{K})=H(W_\theta|{\bf c}_K)
=H(A^{[N]}_\theta|{\bf c}_K,\mathcal{Q})-H(A^{[N]}_\theta|{\bf c}_K, W_\theta,\mathcal{Q})
=H(A^{[N]}_\theta|{\bf c}_K,\mathcal{Q})-H(A^{[N]}_\theta|{\bf c}_{\chi_{\{\theta\}}},{\bf c}_K,\mathcal{Q}).
\end{equation} By recursively using \eqref{eq20} in Lemma \ref{lem4}, \eqref{eq24}, and \eqref{eq010}, we have
\begin{align*}
H(A^{[N]}_\theta|{\bf c}_K,\mathcal{Q})\geq L(\frac{K-1}{K})(1+\sum^{M-1}_{j=1}\prod^j_{i=1}\frac{K-1-it}{N-1-it})
=L(\frac{K-1}{K}) \frac{N-1}{N-1-(M-1)t}.
\end{align*}
 By \eqref{eq23}, one can know that its  retrieval rate
\begin{align*}
R_{\rm PIR}&=\frac{L}{D}=\frac{L}{\sum^{N}_{i=1}H(A^{(i)}_\theta)}
\leq\frac{L}{H(A^{[N]}_\theta|{\bf c}_K,\mathcal{Q})}
\leq \frac{K(N-K+t)}{(K-1)(N-1)}.
\end{align*}\end{proof}
\begin{remark} According to the proof, one can find that a necessary condition for  equality to hold in \eqref{eq22} is that $H(A^{(K)}_\theta)=0$ for every capacity-achieving scheme. However, this conflicts with the design principle for almost all known capacity-achieving non-colluding PIR schemes: the greedy principle, that is, the user downloads data from each server as efficiently as possible with guaranteed privacy. This suggests that the upper bound in \eqref{eq22} may
not be tight in general.
\end{remark}

\section{Constructions for Joint Systematic MDS-Coded PIR Schemes }
 In this section,
 we construct three joint systematic MDS-coded PIR schemes for the cases $N\leq K+t, K=Mt$; $N>K+t, K=Mt$; and $N\leq K+t, K=Mt+1$.
Moreover, for the case $K=Mt, N\leq K+t$, the first construction achieves the upper bound in \eqref{eq14}.
  \subsection{Examples for $(M,N,K=Mt)$}\label{sec4A}

  To illustrate the main idea, we begin with two examples. The first is for the case of $N-t\leq K$ and the second is for the case of $N-t> K$.
\begin{example}\label{ex2} Let $N=5,K=4,M=2,L=6, \ell=3$, i.e., $t=2$ and $N<K+t$. Suppose $W_1=(a_{1,1},a_{1,2},a_{1,3},a_{2,1},a_{2,2},a_{2,3})=({\bf c}_1,{\bf c}_2)$, $W_2=(b_{1,1},b_{1,2},b_{1,3},b_{2,1},b_{2,2},b_{2,3})=({\bf c}_3,{\bf c}_4)\in\mathbb{F}^6_2$.
Moreover, let ${\bf c}_5=\sum^4_{i=1}{\bf c}_i$.
 ${\rm Serv}^{(i)}$ stores the data ${\bf c}_i$ for all $i\in[5]$. Hence, one can directly recover the two files from any four servers.  Without loss of generality, suppose the user wants to retrieve the file $W_1$. The PIR scheme works as follows.

 First, let $\sigma $ be a permutation  of $\{1,2,3\}$ privately chosen by the user  uniformly from  the symmetric group ${\rm S}_3$,  and define $S=({\bf e}^\intercal_{\sigma(1)},{\bf e}^\intercal_{\sigma(2)},{\bf e}^\intercal_{\sigma(3)})\in\mathbb{F}_2^{3\times 3}$, where ${\bf e}_i$ is the $i$th unit vector in $\mathbb{F}_2^3$. We use the random matrix $S$ to  construct the queries as follows:
  $Q^{(1)}_1=S(:,\{1,3\}), Q^{(2)}_1=S(:,\{2,3\}),Q^{(3)}_1=Q^{(4)}_1=Q^{(5)}_1=S(:,\{1,2\}).$
  Then, the answers are determined by $A^{(i)}_1={\bf c}_iQ^{(i)}_1$, which are shown in Figure \ref{fg01}.
\begin{figure}[ht]
	\centering
	\setlength{\abovecaptionskip}{0.2cm}
	\setlength{\belowcaptionskip}{-0.2cm}
	\begin{minipage}[t]{0.48\textwidth}
		\centering
		\resizebox{\linewidth}{!}{
			\renewcommand{\arraystretch}{1.2} 
			\begin{tabular}{|c|c|c|c|c|}
				\hline
				$\rm{Serv}^{(1)}$ & $\rm{Serv}^{(2)}$ & $\rm{Serv}^{(3)}$ & $\rm{Serv}^{(4)}$ & $\rm{Serv}^{(5)}$\\
				\hline
				$a_{1,\sigma(1)}$ &                   & $b_{1,\sigma(1)}$ & $b_{2,\sigma(1)}$ & $a_{1,\sigma(1)}+a_{2,\sigma(1)}+b_{1,\sigma(1)}+b_{2,\sigma(1)}$\\
				& $a_{2,\sigma(2)}$ & $b_{1,\sigma(2)}$ & $b_{2,\sigma(2)}$ & $a_{1,\sigma(2)}+a_{2,\sigma(2)}+b_{1,\sigma(2)}+b_{2,\sigma(2)}$ \\
				$a_{1,\sigma(3)}$ & $a_{2,\sigma(3)}$ &                   &                   & \\
				\hline
			\end{tabular}
		}
		\caption{Answers for retrieving $W_1$}
		\label{fg01}
	\end{minipage}\hfill
	\begin{minipage}[t]{0.48\textwidth}
		\centering
		\resizebox{\linewidth}{!}{
			\renewcommand{\arraystretch}{1.2}
			\begin{tabular}{|c|c|c|c|c|}
				\hline
				$\rm{Serv}^{(1)}$ & $\rm{Serv}^{(2)}$ & $\rm{Serv}^{(3)}$ & $\rm{Serv}^{(4)}$ & $\rm{Serv}^{(5)}$\\
				\hline
				$a_{1,\tau(1)}$   & $a_{2,\tau(1)}$   & $b_{1,\tau(1)}$   &                   & $a_{1,\tau(1)}+a_{2,\tau(1)}+b_{1,\tau(1)}+b_{2,\tau(1)}$\\
				$a_{1,\tau(2)}$   & $a_{2,\tau(2)}$   &                   & $b_{2,\tau(2)}$   & $a_{1,\tau(2)}+a_{2,\tau(2)}+b_{1,\tau(2)}+b_{2,\tau(2)}$ \\
				&                   & $b_{1,\tau(3)}$   & $b_{2,\tau(3)}$   & \\\hline
			\end{tabular}
		}
		\caption{Answers for retrieving $W_2$, where $\tau\sim U({\rm S}_3)$.}
		\label{fg02}
	\end{minipage}
	
\end{figure}

Next, we explore the correctness of the scheme. From Figure \ref{fg01}, the user can  obtain $a_{1,\sigma(1)},a_{1,\sigma(3)},a_{2,\sigma(2)},a_{2,\sigma(3)}$. Moreover,  $a_{1,\sigma(2)}$  and $a_{2,\sigma(1)}$ can be directly recovered from downloaded $4$ symbols in each of the first two rows of Figure \ref{fg01}, respectively.
Hence, the correctness of this scheme is guaranteed.

As to the privacy, we show the answers for respectively retrieving $W_1$ and $W_2$ are identically distributed. We first list the answers for retrieving $W_2$ in Figure \ref{fg02}, where  $\tau \sim U({\rm S}_3)$ means that $\tau$ is uniformly chosen from  the symmetric group ${\rm S}_3$. According to Figures \ref{fg01} and \ref{fg02},  one can find that no matter which file is desired, for each server, the answers are all two independent symbols which are  uniformly chosen from three independent symbols, thus they are identically distributed.

Moreover, the desired file consists of $6$ symbols while the answers totally contain $10$ symbols in Figure \ref{fg01}, so the scheme has
rate $\frac{6}{10}=\frac{3}{5}=\frac{N-K/2}{N}$, which  achieves the upper bound \eqref{eq14} for this case and is strictly higher than the capacity of the separate MDS-coded PIR  $C_\oplus(2,5,4)=\frac{1}{1+4/5}=\frac{5}{9}$.
\end{example}

\begin{example}\label{ex3} Let $N=7,K=4,M=2,L=18, \ell=9$, i.e., $t=2$ and $N>K+t$. Suppose  $W_1=({\bf c}_1,{\bf c}_2),W_2=({\bf c}_3,{\bf c}_4)\in\mathbb{F}^{18}_q$ with $q\geq 7$, and $\tilde{G}$ is a  systematic generator matrix of a
$[7,4]$ MDS code over $\mathbb{F}_q$. Let ${\bf c}_j=(c_{j,1},c_{j,2},\cdots,c_{j,9})\in\mathbb{F}^9_q$ be the data stored in the server ${\rm Serv}^{(j)}$ for $j\in[7]$. Moreover, for $i\in[9]$, $(c_{1,i},c_{2,i},\cdots,c_{7,i})$ is a codeword of $[7,4]$ MDS code by a systematic generator matrix  $\tilde{G}$, i.e., $
({\bf c}_1,\cdots,{\bf c}_7)=({\bf c}_1,\cdots,{\bf c}_4)\tilde{G}\otimes I_{9}.$
Without loss of generality, suppose the user wants to retrieve $W_1$. The PIR scheme works as follows.

 First, let $\sigma $ be a permutation  of $[9]$ privately chosen by the user  uniformly from  the symmetric group ${\rm S}_9$,  and define $S=({\bf e}^\intercal_{\sigma(1)},{\bf e}^\intercal_{\sigma(2)},\cdots,{\bf e}^\intercal_{\sigma(9)})\in\mathbb{F}_q^{9\times 9}$, where ${\bf e}_i$ is the $i$th unit vector in $\mathbb{F}_q^9$. We use the random matrix $S$ to  construct the queries as follows:
  \begin{align*}Q^{(1)}_1&= Q^{(2)}_1=S(:,[4]), Q^{(3)}_1=S(:,[5:9]\setminus\{6\}), Q^{(4)}_1=S(:,[5:9]\setminus\{7\}),\\Q^{(5)}_1&=S(:,[5:9]\setminus\{8\}),Q^{(6)}_1=S(:,[5:9]\setminus\{9\}),Q^{(7)}_1=S(:,[5:9]\setminus\{5\}).
  \end{align*}Then, the answers are determined by $A^{(i)}_1={\bf c}_iQ^{(i)}_1$, which are shown in Figure \ref{fg001}.
\begin{figure}[ht]
	\centering
	\setlength{\abovecaptionskip}{0.2cm}
	\setlength{\belowcaptionskip}{-0.2cm}

	\begin{minipage}[t]{0.48\textwidth}
		\centering
		\resizebox{\linewidth}{!}{
			\renewcommand{\arraystretch}{1.2} 
			\begin{tabular}{|c|c|c|c|c|c|c|}
				\hline
				$\rm{Serv}^{(1)}$ & $\rm{Serv}^{(2)}$ & $\rm{Serv}^{(3)}$ & $\rm{Serv}^{(4)}$ & $\rm{Serv}^{(5)}$ & $\rm{Serv}^{(6)}$ & $\rm{Serv}^{(7)}$\\
				\hline
				$c_{1,\sigma(1)}$ & $c_{2,\sigma(1)}$ &                   &                   &                   &                   & \\
				$c_{1,\sigma(2)}$ & $c_{2,\sigma(2)}$ &                   &                   &                   &                   & \\
				$c_{1,\sigma(3)}$ & $c_{2,\sigma(3)}$ &                   &                   &                   &                   & \\
				$c_{1,\sigma(4)}$ & $c_{2,\sigma(4)}$ &                   &                   &                   &                   & \\
				&                   & $c_{3,\sigma(5)}$ & $c_{4,\sigma(5)}$ & $c_{5,\sigma(5)}$ & $c_{6,\sigma(5)}$ & \\
				&                   &                   & $c_{4,\sigma(6)}$ & $c_{5,\sigma(6)}$ & $c_{6,\sigma(6)}$ & $c_{7,\sigma(6)}$\\
				&                   & $c_{3,\sigma(7)}$ &                   & $c_{5,\sigma(7)}$ & $c_{6,\sigma(7)}$ & $c_{7,\sigma(7)}$\\
				&                   & $c_{3,\sigma(8)}$ & $c_{4,\sigma(8)}$ &                   & $c_{6,\sigma(8)}$ & $c_{7,\sigma(8)}$\\
				&                   & $c_{3,\sigma(9)}$ & $c_{4,\sigma(9)}$ & $c_{5,\sigma(9)}$ &                   & $c_{7,\sigma(9)}$\\
				\hline
			\end{tabular}
		}
		\caption{Answers for retrieving $W_1$}
		\label{fg001}
	\end{minipage}\hfill
	\begin{minipage}[t]{0.48\textwidth}
		\centering
		\resizebox{\linewidth}{!}{
			\renewcommand{\arraystretch}{1.2} 
			\begin{tabular}{|c|c|c|c|c|c|c|}
				\hline
				$\rm{Serv}^{(1)}$ & $\rm{Serv}^{(2)}$ & $\rm{Serv}^{(3)}$ & $\rm{Serv}^{(4)}$ & $\rm{Serv}^{(5)}$ & $\rm{Serv}^{(6)}$ & $\rm{Serv}^{(7)}$\\
				\hline
				&                   & $c_{3,\tau(1)}$   & $c_{4,\tau(1)}$   &                   &                   & \\
				&                   & $c_{3,\tau(2)}$   & $c_{4,\tau(2)}$   &                   &                   & \\
				&                   & $c_{3,\tau(3)}$   & $c_{4,\tau(3)}$   &                   &                   & \\
				&                   & $c_{3,\tau(4)}$   & $c_{4,\tau(4)}$   &                   &                   & \\
				$c_{1,\tau(5)}$   & $c_{2,\tau(5)}$   &                   &                   & $c_{5,\tau(5)}$   & $c_{6,\tau(5)}$   & \\
				& $c_{2,\tau(6)}$   &                   &                   & $c_{5,\tau(6)}$   & $c_{6,\tau(6)}$   & $c_{7,\tau(6)}$\\
				$c_{1,\tau(7)}$   &                   &                   &                   & $c_{5,\tau(7)}$   & $c_{6,\tau(7)}$   & $c_{7,\tau(7)}$\\
				$c_{1,\tau(8)}$   & $c_{2,\tau(8)}$   &                   &                   &                   & $c_{6,\tau(8)}$   & $c_{7,\tau(8)}$\\
				$c_{1,\tau(9)}$   & $c_{2,\tau(9)}$   &                   &                   & $c_{5,\tau(9)}$   &                   & $c_{7,\tau(9)}$\\
				\hline
			\end{tabular}
		}
		\caption{Answers for retrieving $W_2$, where $\tau \sim U({\rm S}_9)$.}
		\label{fg002}
	\end{minipage}
	
\end{figure}

Next, we explore the correctness of the scheme. From Figure \ref{fg001}, the user can directly  obtain $\{c_{1,\sigma(i)},c_{2,\sigma(i)}:i\in[4]\}$. Moreover, the remaining symbols $c_{1,\sigma(i)}, c_{2,\sigma(i)}$ for $i\in[5,9]$ can be recovered from the $4$ symbols in the $i$th row of Figure \ref{fg001}, respectively.
Hence, the correctness of this scheme is guaranteed.

As to the privacy, we show the answers for respectively retrieving $W_1$ and $W_2$ are identically distributed. We first list the answers for retrieving $W_2$ in Figure \ref{fg002}.
According to Figures \ref{fg001} and \ref{fg002},  one can find that no matter which file is desired, for each server, the answers are  $4$ independent symbols, thus they are identically distributed.

Moreover, the desired file consists of 18 symbols, and the total download consists of $4\times 7=28$ symbols in Figure \ref{fg001}. Hence the retrieval rate is $18/28=9/14<1-\frac{1}{2}\times\frac{4}{7}=\frac{5}{7}$. Although the rate has not reached the upper bound \eqref{eq14}, it is strictly higher than the capacity of the separate MDS-coded PIR  $C_\oplus(2,7,4)=\frac{1}{1+4/7}=\frac{7}{11}$, that is, $\frac{9}{14}>\frac{7}{11}$.
\end{example}

\subsection{Construction  for $(M,N, K=Mt)$}

Based on examples in Section \ref{sec4A},  we first present the storage code and then show a formal description of the general joint systematic MDS-coded PIR schemes for the cases $(M,N\leq K+t, K=Mt)$ and $(M,N> K+t, K=Mt)$, respectively.

Let $W_{i}=({\bf c}_{(i-1)t+1},\cdots, {\bf c}_{it})\in \mathbb{F}_q^{L}$ for $i\in[M]$ and the server ${\rm Serv}^{(j)},j\in[N]$ stores the  data ${\bf c}_j=(c_{j,1},c_{j,2},\cdots,c_{j,\ell})$. Moreover, each row $(c_{1,i},c_{2,i},\cdots,c_{N,i}), i\in[\ell]$ is a codeword of  an
$[N,K]$ MDS code $\tilde{\mathcal{C}}$ over $\mathbb{F}_q$ with $q\geq N$. That is, the joint storage  strategy is as follows:
\begin{equation}\label{eq25}({\bf c}_1,\cdots,{\bf c}_N)=({\bf c}_1,\cdots,{\bf c}_K)\tilde{G}\otimes I_{\ell},\end{equation}
where $\tilde{G}$ is  a  systematic generator matrix  of the code $\tilde{\mathcal{C}}$.
The parameters $L$ and $\ell$ will be determined in the following constructions of the joint MDS-coded PIR schemes.

 We first build a scheme for $(M,N\leq K+t, K=Mt)$ in {\bf Construction 1}.

\noindent
\rule{\linewidth}{1pt}
\vspace{-0.3em}
\noindent
\textbf{Construction 1: Joint Systematic MDS-coded PIR Scheme for $(M,N\leq K+t, K=Mt)$}

\vspace{-0.3em}
\noindent
\rule{\linewidth}{0.5pt}

Suppose the storage code  $\mathcal{C}$ is a systematic $(N,K;\ell)$  MDS array code defined by \eqref{eq25} over $\mathbb{F}_q$ with $q \geq N$, and the storage pattern $\mathcal{P}=\ell I_M\otimes {\bf 1}_t$, they are publicly known. Moreover, let $L=t(N-K+t)$ and $\ell=N-K+t$. The goal is to privately retrieve the file $W_\theta$ for any $\theta\in[M]$.
\begin{itemize}
\item \textbf{Preparation:} Let $P=({\bf u}^\intercal_{2},{\bf u}^\intercal_{3},\cdots,{\bf u}^\intercal_{t},{\bf u}^\intercal_{1})\in\mathbb{F}_q^{t\times t}$ and ${\bf v}=({\bf 1}_{K+t-N},{\bf 0}_{N-K})\in \mathbb{F}_q^t$, define $$V=\begin{pmatrix}{\bf v}^\intercal&
 ({\bf v}P)^\intercal&\cdots&({\bf v}P^{t-1})^\intercal\end{pmatrix}^\intercal,$$ where ${\bf u}_{j}$ is the $j$th unit vector in $\mathbb{F}_q^t$.
\item \textbf{Query Phase:} The user privately and uniformly chooses a  permutation $\sigma$ of $[\ell]$  from  the symmetric group ${\rm S}_\ell$. Define the random matrix $S=({\bf e}^\intercal_{\sigma(1)},{\bf e}^\intercal_{\sigma(2)},\cdots,{\bf e}^\intercal_{\sigma(\ell)})\in\mathbb{F}_q^{\ell\times \ell}$, where ${\bf e}_i$ is the $i$th unit vector in $\mathbb{F}_q^\ell$.
Then the user utilizes the random matrix $S$ and public matrix $V$ to  construct the queries as follows:
  \begin{equation*}
  Q^{(j)}_\theta= \begin{cases}S(:,[t]) & {\rm if ~} j\in[N]\setminus[(\theta-1)t+1:\theta t],\\
  S(:,{\rm supp}(V(:,i))\cup[t+1:\ell]) & {\rm if ~} j=(\theta-1)t+i, i\in[t],
  \end{cases}\end{equation*}
  where for a vector ${\bf a}=(a_1,a_2,\cdots,a_n)$,  ${\rm supp}({\bf a})=\{i\in[n]:a_i\neq 0\}$, which is the support of ${\bf a}$.
\item \textbf{Response Phase:} For $\theta\in[M],j\in[N]$,
the answer $A^{(j)}_\theta$ is determined by  \begin{align}
  A^{(j)}_\theta=&{\bf c}_jQ^{(j)}_\theta \notag\\
  =&\begin{cases}\{c_{j,\sigma(x)}:x\in{\rm supp}(V(:,i))\cup[t+1:\ell]\} & {\rm if~} j=(\theta-1)t+i,i\in[t],\\
  \{c_{j,\sigma(x)}:x\in[t]\} & {\rm otherwise.}
  \end{cases} \label{eq26}
  \end{align}
\end{itemize}

\vspace{-0.3em}
\noindent
\rule{1\linewidth}{1pt}
 Now, let us show the correctness and privacy of the scheme in {\bf Construction 1} by the following theorem:
\begin{theorem}\label{thm11} The scheme in {\bf Construction 1} is a joint systematic  MDS-coded PIR scheme  for the case of $(M,N\leq K+t,K=Mt)$ with the   retrieval rate $R_{\rm PIR}=1-\frac{M-1}{M}\cdot \frac{K}{N}$. Consequently, for the storage pattern $\mathcal{P}=\ell I_M\otimes {\bf 1}_t$,
\begin{equation}\label{eq27}
C^{(s)}_{\mbox{\tiny\rm JMDS-PIR}}(M,N,K=Mt,\mathcal{P})= 1-(1-\frac{1}{M})\frac{K}{N}.
  \end{equation}
\end{theorem}
  \begin{proof}We first prove the correctness of  {\bf Construction 1}.  According to $W_\theta=({\bf c}_{(\theta-1)t+1},{\bf c}_{(\theta-1)t+2},\cdots,{\bf c}_{\theta t})$ and \eqref{eq26}, the user can directly  obtain the desired symbols  $\{c_{(\theta-1)t+i,\sigma(x)}:i\in[t],x\in[t+1:\ell]\}$.  By \eqref{eq25}, it has that for each $x\in[t]$,   $(c_{1,\sigma(x)},c_{2,\sigma(x)},\cdots,c_{N,\sigma(x)})$ is a codeword of the $[N,K]$ MDS code  $\tilde{C}$. Hence,  the user needs to know $K$ symbols of such  codeword  for recovering the symbols $\{c_{(\theta-1)t+i,\sigma(x)}:i\in[t]\}$ for $x\in[t]$. Actually, for each $x\in[t]$, the user can  directly download $N-t$ symbols $\{c_{j,\sigma(x)}:j\in[N]\setminus[(\theta-1)t+1:\theta t]\}$  from the servers $\{{\rm Serv}^{(j)}: j\in[N]\setminus[(\theta-1)t+1:\theta t]\}$. Moreover,  for each $x\in[t]$, he also can download  symbols $\{c_{(\theta-1)t+i,\sigma(x)}:x\in {\rm supp}(V(:,i)),i\in[t]\}$ from the remaining $t$
servers.  one can find that $x\in {\rm supp}(V(:,i))$ for $x\in[t]$ means that  $V(x,i)\neq 0$, hence
$$|\{c_{(\theta-1)t+i,\sigma(x)}:x\in {\rm supp}(V(:,i)),i\in[t]\}|=|{\rm supp}(V(x,:))|=K-N+t.$$ So, the user downloads $(N-t)+(K-N+t)=K$ symbols  of the codeword $(c_{1,\sigma(x)},c_{2,\sigma(x)},\cdots,c_{N,\sigma(x)})$ from all servers for each $x \in [t]$. Hence, the correctness of the scheme in {\bf Construction 1}  is guaranteed.

As to the privacy, when the user retrieves the file $W_\theta$, for server ${\rm Serv}^{(j)},j\in[N]\setminus[(\theta-1)t+1:\theta t]$, the answers $A^{(j)}_\theta$ contain $t$ symbols.  For $j\in[(\theta-1)t+1:\theta t]$, the answers $A^{(j)}_\theta$ contains $|{\rm supp}(V(:,j-(\theta-1)t))|+\ell-t$ symbols. According to the definition of $V$, one can find that the Hamming weight of each column in $V$ equals to  that of each row of $V$, i.e., $|{\rm supp}(V(:,i))|=|{\rm supp}({\bf v})|=K+t-N$. Hence, $|{\rm supp}(V(:,j-(\theta-1)t))|+\ell-t=K+t-N+\ell-t=t.$
Therefore, no matter which file is retrieved, each server only knows that the user uniformly and randomly downloaded $t$ symbols from its stored $\ell$ symbols, that is, the answers $A^{(j)}_\theta$ for each $j\in[N]$ are independent of $\theta$. So, the privacy of this scheme is guaranteed.

Moreover, such a scheme downloads a total of $D=tN$ symbols, so its  retrieval rate $R_{\rm PIR}=\frac{L}{D}=\frac{t(t+N-K)}{Nt}=1-\frac{M-1}{M}\cdot \frac{K}{N}$, which achieves the upper bound \eqref{eq14}. Hence, \eqref{eq27} holds for the case of $(M,N\leq K+t, K=Mt)$.
\end{proof}
Next, we build a scheme for the case of $(M,N> K+t, K=Mt)$ in {\bf Construction 2}.

\noindent
\rule{\linewidth}{1pt}
\vspace{-0.3em}
\noindent
\textbf{Construction 2: Joint Systematic MDS-coded PIR Scheme for $(M,N> K+t, K=Mt)$}

\vspace{-0.3em}
\noindent
\rule{\linewidth}{0.5pt}

Suppose the storage code  $\mathcal{C}$ is a systematic $(N,K;\ell)$  MDS array code  defined by \eqref{eq25} over $\mathbb{F}_q$ with $q \geq N$, and the storage pattern $\mathcal{P}=\ell I_M\otimes {\bf 1}_t$, they are publicly known.  Moreover, let $L=t(N+K-t)$ and $\ell=N+K-t$. The goal is to privately retrieve the file $W_\theta$ for any $\theta\in[M]$.
\begin{itemize}
\item \textbf{Preparation:} let $P=({\bf u}^\intercal_{2},{\bf u}^\intercal_{3},\cdots,{\bf u}^\intercal_{N-t},{\bf u}^\intercal_{1})\in\mathbb{F}_q^{(N-t)\times (N-t)}$ and ${\bf v}=({\bf 1}_{K},{\bf 0}_{N-t-K})\in \mathbb{F}_q^{N-t}$, define $$V=\begin{pmatrix}{\bf v}^\intercal&
 ({\bf v}P)^\intercal&\cdots&({\bf v}P^{N-t-1})^\intercal\end{pmatrix}^\intercal,$$ where ${\bf u}^\intercal_{j}$ is the $j$th unit vector in $\mathbb{F}_q^{N-t}$.
\item \textbf{Query Phase:} The user privately and uniformly chooses a  permutation $\sigma$ of $[\ell]$  from  the symmetric group ${\rm S}_\ell$. Define the random matrix $S=({\bf e}^\intercal_{\sigma(1)},{\bf e}^\intercal_{\sigma(2)},\cdots,{\bf e}^\intercal_{\sigma(\ell)})\in\mathbb{F}_q^{\ell\times \ell}$, where ${\bf e}_i$ is the $i$th unit vector in $\mathbb{F}_q^\ell$.
Then the user utilizes the random matrix $S$ and public matrix $V$ to  construct the queries as follows:
  \begin{equation*}\label{eq28}
  Q^{(j)}_\theta= \begin{cases}S(:,{\rm supp}(V(:,j))+K) & {\rm if ~} j\in[(\theta-1)t],\\
  S(:,[K]) & {\rm if ~} j\in[(\theta-1)t+1:\theta t],\\
  S(:,{\rm supp}(V(:,j-t))+K) & {\rm otherwise},
  \end{cases}\end{equation*}
  where for a subset $\Lambda$ of $[n]$, $K+\Lambda=\{K+i:i\in\Lambda\}$.
\item \textbf{Response Phase:} For $\theta\in[M],j\in[N]$,
the answer $A^{(j)}_\theta$ is determined by
\begin{align}
     A^{(j)}_\theta&= {\bf c}_jQ^{(j)}_\theta
  =\begin{cases}\{c_{j,\sigma(x+K)}:x\in{\rm supp}( V(:,j))\} & {\rm if ~} j\in[(\theta-1)t],\\
  \{c_{j,\sigma(x)}:x\in[K]\} & {\rm if ~} j\in[(\theta-1)t+1:\theta t],\\
  \{c_{j,\sigma(x+K)}:x\in{\rm supp}(V(:,j-t))\} & {\rm otherwise}.
  \end{cases} \label{eq28}
\end{align}

\end{itemize}

\vspace{-0.3em}
\noindent
\rule{1\linewidth}{1pt}
We show that the scheme in {\bf Construction 2} satisfies the correctness and privacy conditions as follows.
\begin{theorem}\label{thm13} The scheme in  {\bf Construction 2}  is a joint systematic MDS-coded PIR scheme for the case of $(M,N> K+t,K=Mt)$, and its retrieval rate is $ \frac{(N+K-t)}{NM}$. That is,  for the storage pattern $\mathcal{P}=\ell I_M\otimes {\bf 1}_t$,
\begin{equation*}\label{eq29}
C^{(s)}_{\mbox{\tiny\rm JMDS-PIR}}(M,N>K+t,K=Mt,\mathcal{P})\geq \frac{(N+K-t)}{NM}.
  \end{equation*}
\end{theorem}
\begin{proof} We first prove the correctness of the scheme in {\bf Construction 2}.  According to \eqref{eq28}, the user can directly  download the desired symbols  $\{c_{j,\sigma(x)}:x\in[K]\}$ from the server ${\rm Serv}^{(j)},j\in[(\theta-1)t+1:\theta t]$.  By \eqref{eq25},  a sufficient condition  for recovering the desired symbols $\{c_{(\theta-1)t+i,\sigma(x)}:i\in[t],x\in[K+1:\ell]\}$ is that the user needs to download $K$ symbols of codeword $(c_{1,\sigma(x)},c_{2,\sigma(x)},\cdots,c_{N,\sigma(x)})$ for each $x\in[K+1:\ell]$. Actually, for each $x\in[K+1:\ell]$, the user can  directly download symbols  $$A_x\triangleq\{c_{j,\sigma(x)}:j\in[N]\setminus[(\theta-1)t+1:\theta t],x-K\in {\rm supp}(V(:,\varphi(j)))\},$$
 where $\varphi(j)=j$ for $j\in[(\theta-1)t]$ and $\varphi(j)=j-t$ for $j\in[\theta t+1:N]$.
 Note that $x-K\in {\rm supp}(V(:,\varphi(j)))$ for $x\in[K+1:\ell]$ means that  $V(x-K,\varphi(j))\neq 0$, hence  the size of $A_x$ for $x\in[K+1:\ell]$ is equal to the Hamming weight of the $(x-K)$th row of $V$,
i.e.,
$
|A_x|
=|{\rm supp}(V(x-K,:))|=|{\rm supp}({\bf v}P^{x-K-1})|=|{\rm supp}({\bf v})|=K.$
Hence, the correctness of the scheme in {\bf Construction 2}  is guaranteed.

As to the privacy, it is sufficient to show that each server only knows that the user uniformly and randomly downloaded $K$ symbols  from its stored $\ell$ symbols. By \eqref{eq28},  for server ${\rm Serv}^{(j)},j\in[(\theta-1)t+1:\theta t]$, the answers $A^{(j)}_\theta$ contains $K$ symbols, and for $j\in[N]\setminus [(\theta-1)t+1:\theta t]$, the answers $A^{(j)}_\theta$ contains $|supp(V(:,\varphi(j)))|=K$ symbols,
which is because that
by the definition of $V$, one can find that the Hamming weight of each column in $V$ is equal to  that of each row of $V$, i.e., $|supp(V(:,i))|=|supp({\bf v})|=K$.
Therefore, no matter which file is retrieved, each server only knows that the user uniformly and randomly downloaded $K$ symbols  from its stored $\ell$ symbols. So, the privacy of this scheme is guaranteed.

Moreover, such a scheme downloads a total of $D=KN$ symbols, so its  retrieval rate is $R_{\rm PIR}=\frac{L}{D}=\frac{t(N+K-t)}{NK}=\frac{(N+K-t)}{NM}$. Hence, the desired lower bound on the joint MDS-coded PIR is obtained.
\end{proof}

\subsection{Construction for $(M,N\leq K+t,K=Mt+1)$}
In this subsection, we construct a joint systematic MDS-coded PIR scheme   for $(M,N\leq K+t,K=Mt+1)$, whose retrieval rate is strictly higher than the capacity of the separate MDS-coded PIR. To illustrate the main idea, we begin with an example.
\begin{example}\label{ex4} Let $N=7,K=5,M=2,L=5, \ell=2$. Suppose $W_1=(a_{1,1},a_{1,2},a_{2,1},a_{2,2},a_{3,1})$, $
W_2=(b_{1,1},b_{1,2},b_{2,1},b_{2,2},b_{3,1})\in\mathbb{F}^5_q$ with $q\geq 7$, and $\tilde{G}$ is a  systematic generator matrix of a
$[7,5]$ MDS code over $\mathbb{F}_q$. Suppose ${\bf c}_1=(a_{1,1},a_{1,2}),{\bf c}_2=(a_{2,1},a_{2,2}),{\bf c}_3=(b_{1,1},b_{1,2}),{\bf c}_4=(b_{2,1},b_{2,2}), {\bf c}_5=(a_{3,1},b_{3,1})$, and then $({\bf c}_1,\cdots,{\bf c}_7)=({\bf c}_1,\cdots,{\bf c}_5)\tilde{G}\otimes I_2.$
\begin{figure}[ht]
	\centering
	\setlength{\abovecaptionskip}{0.2cm}
	\setlength{\belowcaptionskip}{-0.2cm}
	\begin{minipage}[t]{0.48\textwidth}
		\centering
		\resizebox{\linewidth}{!}{
			\renewcommand{\arraystretch}{1.2}
			\begin{tabular}{|c|c|c|c|c|c|c|}
				\hline
				$\rm{Serv}^{(1)}$ & $\rm{Serv}^{(2)}$ & $\rm{Serv}^{(3)}$ & $\rm{Serv}^{(4)}$ & $\rm{Serv}^{(5)}$ & $\rm{Serv}^{(6)}$ & $\rm{Serv}^{(7)}$\\\hline
				$a_{1,\sigma(1)}$& $a_{2,\sigma(1)}$&&& $a_{3,1}$&&\\ 
				& &$b_{1,\sigma(2)}$& $b_{2,\sigma(2)}$& $b_{3,1}$& $c_{\sigma(2)}$ & $d_{\sigma(2)}$\\
				\hline
			\end{tabular}
		}
		\caption{Answers for retrieving $W_1$, where $\sigma\sim U({\rm S}_2)$}
		\label{fg03}
	\end{minipage}\hfill
	\begin{minipage}[t]{0.48\textwidth}
		\centering
		\resizebox{\linewidth}{!}{
			\renewcommand{\arraystretch}{1.2}
			\begin{tabular}{|c|c|c|c|c|c|c|}
				\hline
				$\rm{Serv}^{(1)}$& $\rm{Serv}^{(2)}$ & $\rm{Serv}^{(3)}$ & $\rm{Serv}^{(4)}$ & $\rm{Serv}^{(5)}$ & $\rm{Serv}^{(6)}$ & $\rm{Serv}^{(7)}$\\
				\hline
				&&$b_{1,\tau(1)}$& $b_{2,\tau(1)}$& $a_{3,1}$&&\\ 
				$a_{1,\tau(2)}$& $a_{2,\tau(2)}$& && $b_{3,1}$& $c_{\tau(2)}$ & $d_{\tau(2)}$\\
				\hline
			\end{tabular}
		}
		\caption{Answers for retrieving $W_2$, where $\tau\sim U({\rm S}_2)$.}
		\label{fg04}
	\end{minipage}
	
\end{figure}
The PIR scheme is illustrated by the answers for retrieving $W_1$ in Figure   \ref{fg03} and for retrieving  $W_2$ in Figure \ref{fg04}, respectively.

Next, we explore the correctness and the privacy of the scheme. The correctness follows from the observation that one can directly obtain three desired symbols and the remaining $2$ desired symbols can be recovered from the other five symbols in a codeword, which is because that such $7$ symbols consist of a codeword of $[7, 5]$ MDS code. Moreover, the privacy of this scheme is guaranteed because that the data stored in ${\rm Serv}^{(5)}$ is fully downloaded and  no matter which file is desired, for each of the other servers, the answers are identically distributed.

Moreover, the desired file consists of $5$ symbols while the answers totally contain $8$ symbols, so the scheme has
rate $\frac{5}{8}=\frac{7-2}{7+1}$, which is strictly higher than the capacity of the separate MDS-coded PIR  $C_\oplus(2,7,5)=\frac{1}{1+5/7}=\frac{7}{12}$.
\end{example}

Now we show a formal description of the general scheme.

\noindent
\rule{\linewidth}{1pt}
\vspace{-0.3em}
\noindent
\textbf{Construction 3: Joint  Systematic MDS-coded PIR Scheme for $(M,N\leq K+t, K=Mt+1)$}

\vspace{-0.3em}
\noindent
\rule{\linewidth}{0.5pt}
Suppose the storage pattern $\mathcal{P}=\frac{L}{K}(MI_M\otimes {\bf 1}_t~{\bf 1}_M^{\intercal})$, they are publicly known.
The goal is to privately retrieve the file $W_\theta$ for any $\theta\in[M]$.
\begin{itemize}
\item \textbf{Storage System:}
The joint storage  strategy is presented as follows.
Suppose $W_{i}=({\bf c}_{(i-1)t+1},\cdots, {\bf c}_{it},{\bf W}_{i,K})\in \mathbb{F}_q^{L}$ for $i\in[M]$ and ${\bf c}_K=({\bf W}_{1,K},{\bf W}_{2,K},\cdots,{\bf W}_{M,K})\in\mathbb{F}_q^\ell$. The server ${\rm Serv}^{(j)},j\in[N]$ stores the  data ${\bf c}_j$  which is coded by a systematic generator matrix $\tilde{G}$ of an
$[N,K]$ MDS code $\tilde{C}$ over $\mathbb{F}_q$ with $q\geq N$, that is,
 $({\bf c}_1,\cdots,{\bf c}_N)=({\bf c}_1,\cdots,{\bf c}_K)\tilde{G}\otimes I_{\ell}.$
 According to the above encoding process, we know that each row $(c_{1,j},c_{2,j},\cdots,c_{N,j}), j\in[\ell]$ is a codeword of the MDS code $\tilde{C}$.
 \item \textbf{Preparation:}
 Let $L=\frac{K\ell}{M}$,   $\ell=\mu(N-K+t)$, and $\mu=\frac{M}{\gcd(N-K+t,M)}$.  Moreover, let $P=({\bf u}^\intercal_{2},{\bf u}^\intercal_{3},\cdots,{\bf u}^\intercal_{t},{\bf u}^\intercal_{1})\in\mathbb{F}_q^{t\times t}$ and ${\bf v}=({\bf 1}_{K+t-N},{\bf 0}_{N-K})\in \mathbb{F}_q^t$, define $V=({\bf v}^\intercal~
 ({\bf v}P)^\intercal ~\cdots~ ({\bf v}P^{t-1})^\intercal)^\intercal,$  where ${\bf u}^\intercal_{j}$ is the $j$th unit vector in $\mathbb{F}_q^t$.
 \item \textbf {Query Phase:}
 The user  privately and uniformly chooses  a permutation $\sigma $  from  the symmetric group ${\rm S}_\ell$,  and define the random matrix  $S=({\bf e}^\intercal_{\sigma(1)},{\bf e}^\intercal_{\sigma(2)},\cdots,{\bf e}^\intercal_{\sigma(\ell)})\in\mathbb{F}_q^{\ell\times \ell}$, where ${\bf e}_i$ is the $i$th unit vector in $\mathbb{F}_q^\ell$.
 Then, he uses the random matrix $S$ and public matrix $V$ to construct the queries as follows:
  for $j\in[N]$,{\small\begin{equation}\label{eq30}
  Q^{(j)}_\theta= \begin{cases}
  I_\ell  & {\rm if ~} j=K,\\
  S(:,\cup_{i=0}^{\mu-1}(it+{\rm supp}(V(:,j')))\cup[\mu t+1:\ell]) & {\rm if ~} j=(\theta-1)t+j',j'\in[ t],\\
  S(:,[\mu t]) & {\rm otherwise}.
\end{cases}\end{equation}}
\item \textbf{Response Phase:} For $\theta\in[M],j\in[N]$,
the answer $A^{(j)}_\theta$ is determined by
{\small\begin{equation}\label{eq30'}
  A^{(j)}_\theta
  =\begin{cases}
  {\bf c}_K  & {\rm if ~} j=K,\\
  \cup^{\mu-1}_{i=0}\{c_{j,\sigma(x+it)}:x\in {\rm supp}(V(:,j'))\}\cup\{c_{j,\sigma(x)}:x\in[\mu t+1:\ell]\} & {\rm if ~} j=(\theta-1)t+j',j'\in[t],\\
  \{c_{j,\sigma(x)}:x\in[\mu t]\} & {\rm otherwise}.
  \end{cases}
\end{equation}}
\end{itemize}
\noindent
\rule{\linewidth}{1pt}

 We first explain the privacy of the scheme in {\bf Construction 3}. One can find that $A^{(K)}_\theta={\bf c}_K$, which implies that the ${\rm Serv}^{(K)}$ gets no information of $\theta$. By \eqref{eq30} and \eqref{eq30'}, one can know that for $j\in[(\theta-1)t+1:\theta t]$, the user uniformly and randomly downloads $\mu |{\rm supp}(V(:,j-(\theta-1)t))|+\ell-\mu t=\mu (K-N+t)+\mu(N-K)=\mu t$ symbols from ${\rm Serv}^{(j)}$, so these servers also get no information of $
  \theta$. Moreover, for the remaining servers, one can find that the number of  symbols downloaded  uniformly and randomly from  each of them is  $\mu t$, which implies that  each  of them  also gets no information of $\theta$. So, the privacy of this scheme is guaranteed.

  Next, let us explain the correctness of the scheme in {\bf Construction 3}. By \eqref{eq30'}, the user can directly download  the desired symbols $\{c_{j,\sigma(x)} : j\in[(\theta-1)t+1:\theta t],x\in [\mu t+1:\ell]\}\cup\{{\bf W}_{\theta,K}\}$ for $j\in[(\theta-1)t+1:\theta t]$. Hence, the user only needs to recover the desired symbols $\{\{c_{j,\sigma(x)} : j\in[(\theta-1)t+1:\theta t],x\in [\mu t]\}\}$. According to the fact that  $(c_{1,\sigma(x)},c_{2,\sigma(x)},\cdots, c_{N,\sigma(x)})$  for $x\in[\mu t]$ is a  codeword of MDS code, it is sufficient to show that the user can download  $K$ symbols of such codewords. Actually, for each $x\in [\mu t]$ with $x=it+x',i\in[0:\mu-1],x'\in [t]$, the user directly downloads $N-t$ symbols from the servers labeled by $[N] \setminus [(\theta-1)t+1:\theta t]$  and downloads
  $$|\{c_{(\theta-1)t+j',\sigma(x)}:x'\in {\rm supp}(V(:,j')),j'\in[t]\}|=|{\rm supp}(V(x',:))|= |{\rm supp}({\bf v})|=K-N+t
$$
  symbols from the servers  labeled by $[(\theta-1)t+1:\theta t]$, that is, the user downloads $N-t+K-N+t=K$ symbols of the codeword $(c_{1,\sigma(x)},c_{2,\sigma(x)},\cdots, c_{N,\sigma(x)})$ for $x\in[\mu t]$. Hence, the correctness of this scheme is guaranteed.

 Moreover, such a scheme downloads a total of $D=(N-1)\mu t+\ell=\mu ((N-1)t+N-K+t)$ symbols, so its  retrieval rate $R_{\rm PIR}=\frac{L}{D}=\frac{\mu \frac{K (N-K+t)} {M}}{\mu ((N-1)t+N-K+t) }=\frac{K(N-K+t)}{M((t+1)N-K)}$, which is  always strictly higher than the capacity of separate
MDS-coded PIR schemes  by Proposition~2(i).

\subsection{Comparison}
In this subsection, we compare the proposed joint MDS-coded PIR schemes with known MDS-coded PIR schemes in terms of the {\it retrieval rate $R_{\rm PIR}$}, as summarized in Table \ref{tab11}.

 \begin{proposition}\label{lem15}For $K=Mt, 0<r=\frac{K}{N}<1, M\geq 2$, let $f_M(r)\triangleq\begin{cases}
 1-(1-\frac{1}{M})r-\frac{1-r}{1-r^M}, & {\rm if ~}\frac{M}{M+1}\leq r< 1\\
 \frac{M+(M-1)r}{M^2}-\frac{1-r}{1-r^M}, & {\rm otherwise }
 \end{cases},$ and $\Delta_M(r)\triangleq\frac{f_M(r)}{\frac{1-r}{1-r^M}}
    =(1-\frac{1}{M})(1-r^M)+\frac{1-r^M}{M(1-r)}-1$ for $\frac{M}{M+1}\leq r< 1$.
 Then,
 \begin{itemize}
 \item [(i)] For $\frac{M}{M+1}\leq r\leq 1$ with $M\geq 2$, $\Delta_M(r)$ is  monotonically decreasing with respect to $r$,  and
     \begin{equation}\label{eq32}
     \lim_{M\to\infty} \max_{{M}/(M+1)\leq r\leq 1}\Delta_M(r)=1-\frac{2}{\it e}\approx 26.42\%.
     \end{equation}
     \item [(ii)] For $0<r<\frac{M}{M+1}$,  $f_M (r)$ is monotonically increasing.
And there exists a unique $\frac{M}{M+1}>r_M>0$ such that  $f_M(r_M)=0$, and $f_M(r)>0$ for $r_M<r< 1$.
 \end{itemize}
 \end{proposition}
\begin{proof} The proof of the proposition is divided into two parts.
 \begin{itemize}
 \item[(i)]We first prove that $\Delta_M(r)$ is monotonically decreasing for $\frac{M}{M+1}\leq r<1$. For $\frac{M}{M+1}< r < 1$, \begin{align*}
    \frac{{\rm d}\Delta_M(r)}{{\rm d} r}
    =\frac{1+(M-1)r^{M}-Mr^{M-1}-M(M-1)r^{M-1}(1-r)^2)}{M(1-r)^2},
    \end{align*} hence it is sufficient to show  that  for $\frac{M}{M+1}\leq r \leq 1$, $F(r)\triangleq1+(M-1)r^{M}-Mr^{M-1}-M(M-1)r^{M-1}(1-r)^2)<0$. One can find that $
    \frac{{\rm d}F(r)}{{\rm d} r}=-M(M-1)r^{M-2}(1-r)(M-(M+1)r)\geq 0$, i.e., $F(r)$ is monotonically increasing  for $\frac{M}{M+1}\leq r \leq 1$.
     Hence, it always has that $F(r)<F(1)=0$ for $\frac{M}{M+1}\leq r< 1$, that is, $\frac{{\rm d}\Delta_M(r)}{{\rm d} r}=\frac{F(r)}{M(1-r)^2}<0.$
    So, $\max_{\frac{M}{M+1}\leq r\leq 1,}\Delta_M(r)=\Delta_M(\frac{M}{M+1})=2(1-(\frac{M}{M+1})^M)-1$, and
 $\lim_{M\to\infty}\max_{\frac{M}{M+1}\leq r\leq 1,}\Delta_M(r)=1-\frac{2}{e}.$
\item[(ii)] By (i), we have that $f_M(r)> 0$ for $\frac{M}{M+1}\leq r\leq 1$. Next we prove that for  $0<r< \frac{M}{M+1}$, $f_M(r)$ is monotonically increasing. For $M=2$, $f_2(r)=\frac{2+r}{4}-\frac{1}{1+r}$, which always holds.  For  $0<r< \frac{M}{M+1},M\geq 3$, it has that \begin{align*}
    M^2(r^M-1)^2\frac{{\rm d}f_M(r)}{{\rm d} r}
    =&(M-1)(r^{M}-1)^2+r^{M-1}((M^3-M^2)r-M^3)+M^2\\
    \stackrel{(a)}{\geq}&(M-1)( \frac{M}{M+1})^{2M}-2(M^2+M-1)(\frac{M}{M+1})^M+(M^2+M-1)\\
    =&(M-1)(\frac{M+1}{M})^2( \frac{M}{M+1})^{2(M+1)}-2(M^2+M-1)(\frac{M+1}{M})(\frac{M}{M+1})^{M+1}+(M^2+M-1)\\
    \stackrel{(b)}{\geq}& (M-1)(\frac{M+1}{eM})^2+(M^2+M-1)(1-2\frac{M+1}{eM})
    >(M^2+M-1)(1-2\frac{M+1}{eM})\\
    \geq &(M^2+M-1)(1-\frac{8}{3e})>0,
    \end{align*}
    where  $(a)$ follows from that for $0<r\leq \frac{M}{M+1}$, $\frac{{\rm d}((M^2(r^M-1)^2))\frac{{\rm d}f_M(r)}{{\rm d} r}}{{\rm d} r}=-M(M-1)r^{M-2}(2r(1-r^M)+M^2(1-r))<0,$ and $(b)$ is due to that $(\frac{M}{M+1})^{M+1}<\frac{1}{e}$ and $(M-1)(\frac{M+1}{M})^2x^2-2(M^2+M-1)\frac{M+1}{M}x+(M^2+M-1)$ is monotonically decreasing for $x<\frac{(M^2+M-1)M}{M^2-1}$.
Moreover, $\lim_{r\to 0^{+}}f_M(r)=\frac{1}{M}-1<0$ and $f_M(\frac{M}{M+1})>0$. By the intermediate value property and  monotonicity of $f_M(r)$, there exists a unique root of $f_M(r)$ for $r\in (0,\frac{M}{M+1})$, denoted by $r_M$, and $f_M(r)>0$ for $r_M<r<\frac{M}{M+1}$.
\end{itemize}
\end{proof}

 \begin{proposition}\label{lem16}
 For $K=Mt+1, K< N\leq K+t$ and $M\geq 2$,  let
 $g_{M}(N,t)\triangleq\frac{K(N-K+t)}{M((t+1)N-K)}-\frac{1-K/N}{1-(K/N)^M}$
 and $\delta_M(N,t)\triangleq\frac{g_M(N,t)}{\frac{1-K/N}{1-(K/N)^M}}$. Then,
 \begin{itemize}
 \item [(i)] $g_{M}(N,t)> 0.$
     \item [(ii)] Let $a\mid t$ and $N=K+\frac{t}{a}$, it has that $\lim_{M\to\infty}\lim_{t\to\infty}\delta_{M}(N,t)=a-(1+a)e^{-1/a}\leq 1-\frac{2}{e}\approx 26.42\%.$
 \end{itemize}
 \end{proposition}
 \begin{proof} \begin{itemize}
\item [(i)]
It is equivalent to show that  $\frac{M((t+1)N-K)}{K(N-K+t)} < \frac{1-r^M}{1-r}=\sum_{i=0}^{M-1}r^i,$ where $r=\frac{K}{N}$.
 Note that for $0<r<1$, $\frac{{\rm d}(\sum_{i=0}^{M-1}r^i-M+\binom{M}{2}(1-r))}{{\rm d }r}=\sum^{M-1}_{i=1}i(r^{i-1}-1)<0,$
 thus, $\sum_{i=0}^{M-1}r^i> M-\binom{M}{2}(1-r).$
 Moreover, for $K<N\leq K+t$,
$
M-\binom{M}{2}(1-r)-M\frac{t+1-r}{K(1-r)+tr}
\geq \binom{M}{2}(1-r)^2\frac{(N-K+t)}{K(1-r)+tr}>0.
$
Thus, for $K< N\leq K+t$, it has that
$
\frac{M((t+1)N-K)}{K(N-K+t)}=M\frac{t+1-r}{K(1-r)+tr}< M-\binom{M}{2}(1-r)< \frac{1-r^M}{1-r}.
$
\item [(ii)] Note that $\delta_M(K+\frac{t}{a},t)=(1+a)(1-\frac{1+\frac{1}{Mt}}{(1+\frac{1}{Ma}+\frac{1}{Mt})^{M}})-1$, thus \begin{align*}
&\lim_{M\to\infty}\lim_{t\to\infty}\delta_M(K+\frac{t}{a},t)
=\lim_{M\to\infty}(1+a)(1-\frac{1}{(1+\frac{1}{Ma})^{M}})-1=a-(1+a)e^{-a^{-1}}.
\end{align*}
Moreover, $\frac{{\rm d}(a-(1+a)e^{-a^{-1}})}{{\rm d}a}<0$ for $a\geq 1$, therefore $0<a-(1+a)e^{-1/a}\leq 1-\frac{2}{e}$.
\end{itemize}

\end{proof}
\begin{figure}[ht]
	\centering
	\setlength{\abovecaptionskip}{0.2cm}
	\setlength{\belowcaptionskip}{-0.2cm}
	
	\begin{minipage}[t]{0.45\textwidth} 
		\centering
		\includegraphics[width=\linewidth]{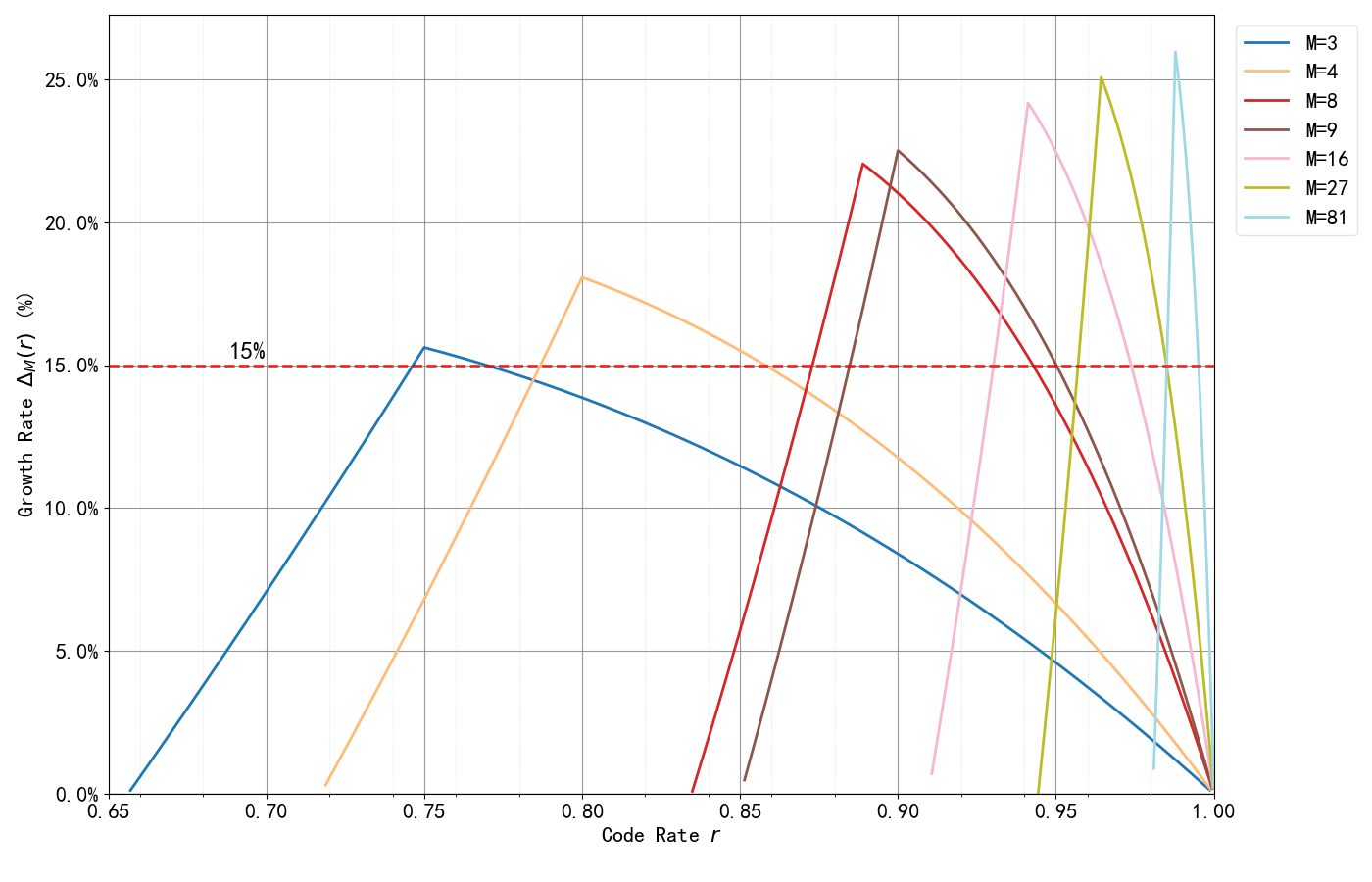}
		\caption{The growth rate $\Delta_M(r)$ for $K=Mt$ and different $M$.}
		\label{fig7}
	\end{minipage}
	\begin{minipage}[t]{0.45\textwidth} 
		\centering
		\includegraphics[width=\linewidth]{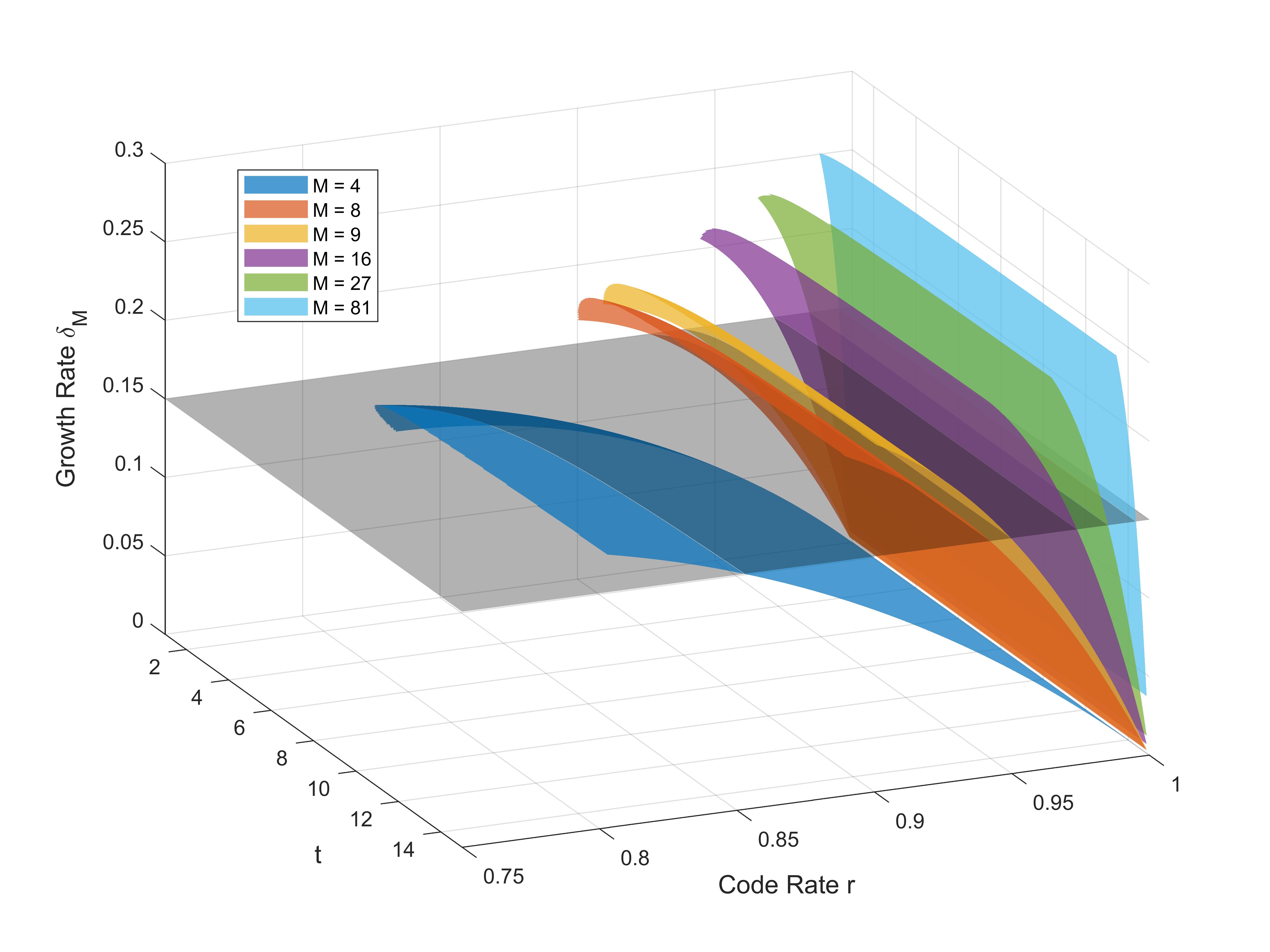}
		\caption{The growth rate $\delta_M(N,t)$ for $K=Mt+1$, $K< N\leq K+t$ and different $M$, where $\delta_M=\delta_M(\frac{Mt+1}{r},t)$ and $\frac{Mt+1}{(M+1)t+1}\leq r<1$.}
		\label{fig8}
	\end{minipage}
	
\end{figure}

\begin{remark}
The functions introduced in Propositions~\ref{lem15} and~\ref{lem16} quantify the performance gain of our joint MDS-coded PIR schemes over capacity-achieving separate MDS-coded PIR schemes. More precisely, for the case $K=Mt$, the function $f_M(r)$ measures the difference between the retrieval rate of our schemes and the separate MDS-coded PIR capacity $C_{\oplus}$ at the same storage code rate $r=K/N$, while $\Delta_M(r)$ measures the corresponding relative improvement. By Proposition~\ref{lem15}(ii), we have $\{0<r<1: f_M(r)>0\}=(r_M,1)$.
Therefore, the scheme in \textbf{Construction~1} for $K=Mt$ and $N\le K+t$ always achieves a retrieval rate strictly larger than $C_{\oplus}$, and the scheme in \textbf{Construction~2} for $K=Mt$ and $K+t<N<K/r_M$ also achieves a retrieval rate strictly larger than $C_{\oplus}$. Moreover, Proposition~\ref{lem15}(i) gives $\max_{\frac{M}{M+1}\le r<1}\Delta_M(r)
    =
    \Delta_M\left(\frac{M}{M+1}\right)
    =
    2\left(1-\left(\frac{M}{M+1}\right)^M\right)-1$, which is monotonically increasing with respect to $M$ and tends to $1-2/e\approx 26.42\%$ as $M$ increases.
For the case $K=Mt+1$, the functions $g_M(N,t)$ and $\delta_M(N,t)$ play the same roles: $g_M(N,t)$ measures the difference between the retrieval rate of \textbf{Construction~3} and $C_{\oplus}$, while $\delta_M(N,t)$ measures the corresponding relative improvement. By Proposition~\ref{lem16}(i), we have $g_M(N,t)>0$,
or equivalently, $R_{\rm PIR}
    =
    \frac{K(N-K+t)}{M((t+1)N-K)}
    >
    C_{\oplus}$,
for $K=Mt+1$ and $K<N\le K+t$. Hence, the scheme in \textbf{Construction~3} always outperforms capacity-achieving separate MDS-coded PIR schemes in this parameter regime. Furthermore, Proposition~\ref{lem16}(ii) shows that the relative improvement can asymptotically reach $1-2/e$. As illustrated in Figures~\ref{fig7} and~\ref{fig8}, the relative improvement of our schemes can exceed $15\%$ when $M\ge 4$, exceed $20\%$ when $M\ge9$, and approach $1-2/e\approx 26.42\%$ as $M$ grows.
\end{remark}

\section{Conclusion}
In this paper, we investigated joint MDS-coded PIR with systematic MDS array storage codes under prescribed storage patterns. We first introduce the storage pattern and MDS array codes to characterize the joint systematic MDS-coded storage strategy, and then clarify the capacity of the joint systematic  MDS-coded PIR model. Next, we derive upper bounds \eqref{eq14} and \eqref{eq22} on the capacity of the joint systematic  MDS-coded PIR schemes with the given storage pattern for the cases of $K=Mt$ and $K=Mt+1$, respectively. To determine the capacity, we construct three classes of joint MDS-coded PIR schemes for the cases $N\leq K+t, K=Mt$; $N>K+t, K=Mt$, and  $K=Mt+1$, respectively. In particular, the scheme for the case of $N\leq K+t,K=Mt$ matches the upper bound  \eqref{eq14}, which implies that the capacity of the joint systematic  MDS-coded PIR scheme for $N\leq K+t, K=Mt$ under the given storage pattern is fully characterized. We also showed that the proposed schemes achieve higher retrieval rates than capacity-achieving separate MDS-coded PIR schemes in the corresponding parameter regimes.
However, the parameter region of joint systematic  MDS-coded PIR schemes considered in this work is still relatively limited. Extending the results to more general joint MDS-coded PIR settings remains an interesting and challenging direction for future research.




\end{document}